\documentclass[a4paper,noarxiv,onecolumn, accepted=2022-06-20]{quantumarticle}
\usepackage[utf8]{inputenc}
\usepackage[english]{babel}
\usepackage[T1]{fontenc}

\usepackage{lipsum}
\usepackage{mdframed}
\usepackage{booktabs}
\usepackage{longtable}

\newcommand*{\ditto}{--- \raisebox{-0.5ex}{''} ---}

\newenvironment{options}
	{\medskip\noindent\begin{longtable}{p{.25\columnwidth}p{.694\columnwidth}}
	\textsf{Option} & \textsf{Description} \\
	\midrule
	}
	{\bottomrule\end{longtable}}

\newcommand{\option}[2]{
	\small\texttt{#1} & {\small#2} \\
}

\newcommand{\compatibilityoption}[2]{
	\small\texttt{#1 (C)} & {\small #2} \\
}

\newcommand{\defaultoption}[2]{
\small\texttt{#1 (default)} & {\small #2} \\
}

\newenvironment{commands}
{\medskip\noindent\begin{longtable}{p{.20\columnwidth}p{.744\columnwidth}}
		\textsf{Command} & \textsf{Description} \\
		\midrule
	}
	{\bottomrule\end{longtable}}

\newcommand{\command}[2]{
	\small\texttt{#1} & {\small#2} \\
}

\usepackage{bbm}
\usepackage{graphicx}
\usepackage{upref}
\usepackage{enumerate}
\usepackage{latexsym}
\usepackage{color,graphics}
\usepackage{comment} 
\usepackage{relsize}
\usepackage{doi} 
\usepackage{hyperref} 

\usepackage[section]{algorithm} 
\usepackage{amsmath,amsthm,amssymb,algorithmic,epsfig,graphicx, tikz}

\usepackage{varioref}

\newtheorem{theorem}{Theorem}[section]
\newtheorem{lemma}[theorem]{Lemma}

\newtheorem{proposition}[theorem]{Proposition}
\newtheorem{corollary}[theorem]{Corollary}

\newtheorem{defn}[theorem]{Definition}

\newcommand{\R}{\mathbb{R}}
\newcommand{\Z}{\mathbb{Z}}

\newcommand{\E}{\mathbb{E}}

\def\bra#1{\mathinner{\langle{#1}|}}
\newcommand{\ket}[1]{\mathinner{|{#1}\rangle}}
\newcommand{\braket}[2]{\langle #1|#2\rangle}

\renewcommand{\part}[2]{\frac{\partial #1}{\partial #2}}

\newcommand{\all}[2]{\begin{align}\label{#2} #1\end{align}}
\newcommand{\al}[1]{\begin{align} #1\end{align}}
\newcommand{\eq}[1]{\begin{equation}#1\end{equation}}

\newcommand{\en}[1]{\left ( #1 \right )}

\newcommand{\nl}{\notag \\}

\newcommand{\norm}[1]{\lVert#1\rVert}

\newcommand{\N}{\mathbb{N}}

\long\def\rem#1{}

\newcommand{\thmref}[1]{\hyperref[#1]{{Theorem~\ref*{#1}}}}
\newcommand{\lemref}[1]{\hyperref[#1]{{Lemma~\ref*{#1}}}}
\newcommand{\remref}[1]{\hyperref[#1]{{Remark~\ref*{#1}}}}
\newcommand{\corref}[1]{\hyperref[#1]{{Corollary~\ref*{#1}}}}
\newcommand{\eqnref}[1]{\hyperref[#1]{{Equation~(\ref*{#1})}}}
\newcommand{\claimref}[1]{\hyperref[#1]{{Claim~\ref*{#1}}}}
\newcommand{\remarkref}[1]{\hyperref[#1]{{Remark~\ref*{#1}}}}
\newcommand{\propref}[1]{\hyperref[#1]{{Proposition~\ref*{#1}}}}
\newcommand{\factref}[1]{\hyperref[#1]{{Fact~\ref*{#1}}}}
\newcommand{\defref}[1]{\hyperref[#1]{{Definition~\ref*{#1}}}}
\newcommand{\exampleref}[1]{\hyperref[#1]{{Example~\ref*{#1}}}}
\newcommand{\hypref}[1]{\hyperref[#1]{{Hypothesis~\ref*{#1}}}}
\newcommand{\secref}[1]{\hyperref[#1]{{Section~\ref*{#1}}}}
\newcommand{\chapref}[1]{\hyperref[#1]{{Chapter~\ref*{#1}}}}
\newcommand{\apref}[1]{\hyperref[#1]{{Appendix~\ref*{#1}}}}

\AtBeginDocument{}

\begin{document}
\title{Low depth algorithms for quantum amplitude estimation}

\author[c,b]{Tudor Giurgica-Tiron}
\thanks{tgt@stanford.edu} 
\author[a,e]{Iordanis Kerenidis}
\thanks{ iordanis.kerenidis@qcware.com} 
\author[d,b]{Farrokh Labib}
\thanks{ farrokhlabib@gmail.com} 
\author[a]{Anupam Prakash}
\thanks{ anupam.prakash@qcware.com} 
\author[b]{William Zeng} 
\thanks{ william.zeng@gs.com}

\affil[a]{QC Ware Corp., Palo Alto, USA and Paris, France.}
\affil[b]{Goldman, Sachs \& Co., New York, USA.}
\affil[c]{Stanford University, Palo Alto, USA.}
\affil[d]{CWI Amsterdam, Netherlands.}
\affil[e]{CNRS, Universit\'{e} Paris, France.}

\maketitle

	\begin{abstract} 
We design and analyze two new low depth algorithms for amplitude estimation (AE) achieving an optimal tradeoff between the quantum speedup and circuit depth. 
For $\beta \in (0,1]$, our algorithms require $N= \tilde{O}( \frac{1}{ \epsilon^{1+\beta}})$ oracle calls and require the oracle to be called sequentially $D= O( \frac{1}{ \epsilon^{1-\beta}})$ times 
to perform amplitude estimation within additive error $\epsilon$. These algorithms interpolate between the classical algorithm $(\beta=1)$ and the 
standard quantum algorithm ($\beta=0$) and achieve a tradeoff $ND= O(1/\epsilon^{2})$. These algorithms bring quantum speedups for Monte Carlo methods 
closer to realization, as they can provide speedups with shallower circuits. 

The first algorithm (Power law AE) uses power law schedules in the framework introduced by Suzuki et al \cite{S20}. The algorithm works for 
$\beta \in (0,1]$ and has provable correctness guarantees when the log-likelihood function satisfies regularity conditions required for the Bernstein Von-Mises theorem. The second algorithm (QoPrime AE) uses the Chinese remainder theorem for combining lower depth estimates to 
achieve higher accuracy. The algorithm works for discrete $\beta =q/k$ where $k \geq 2$ is the number of distinct coprime 
moduli used by the algorithm and $1 \leq q \leq k-1$, and has a fully rigorous correctness proof. We analyze both algorithms in the presence of depolarizing noise and provide numerical comparisons with the state of the art amplitude estimation algorithms. 
\end{abstract} 

\newpage

\section{Introduction}
Amplitude estimation \cite{BHMT02} is a fundamental quantum algorithm that allows a quantum 
computer to estimate the amplitude $\bra{0} U\ket{0}$ for a quantum circuit 
$U$ to additive error $\epsilon$ with $O(1/\epsilon)$ calls to $U$. The algorithm offers a 
quadratic advantage over classical sampling and has many applications including 
speedups for Monte Carlo methods \cite{M15} and approximate counting \cite{BHT98}.

To be more precise, we consider an amplitude estimation setting where the algorithm is given access to a quantum circuit $U$ such that 
$U\ket{0^t} = \cos (\theta) \ket{x, 0} + \sin(\theta) \ket{x', 1}$ where $\ket{x}, \ket{x'}$ are arbitrary states on $(t-1)$ qubits. The algorithm's goal
is to estimate the amplitude $\theta$ within an additive $\epsilon$. The closely related approximate counting problem corresponds to the special case 
where $U\ket{0^t}= \frac{1}{2^{(t-1)/2}}( \sum_{i \in S} \ket{i}\ket{0} + \sum_{i \not\in S} \ket{i}\ket{1} )$ is a uniform superposition over bit strings of length $(t-1)$ and the binary label on the second register is the indicator function for $S \subseteq \{0, 1\}^{t-1}$. Amplitude estimation in this setting provides an estimate for $|S|$. Approximate counting in turn generalizes Grover's search \cite{G98} and the problem of finding the number of marked elements in a list. 

We briefly discuss two applications of amplitude estimation, to quantum Monte Carlo methods and inner product estimation that are particularly relevant for applications of quantum computing to finance and machine learning. 

Quantum Monte Carlo methods are an important application of amplitude estimation to the problem of estimating the mean of a real valued random variable by sampling. Let $f(x,S)$ be a real valued function where $x$ is the input and $S$ is the random seed and let $\sigma$ be the variance of $f(x,S)$. Classically estimating the mean $E_{S} f(x,S)$ to additive error $\epsilon$ is known to require $N= O(\sigma^{2}/\epsilon^{2})$ samples. Montanaro showed that there is a quantum algorithm for estimating the mean that required $N=\widetilde{O}\left(\frac{\sigma}{\epsilon}\right)$ samples, thus quadratically improving the dependence on $\sigma$ and $1/\epsilon$ \cite{M15}. This quantum Monte-Carlo algorithm builds on prior works for estimating the mean of real valued functions using quantum computers \cite{A99,G98,B11} and has been further generalized to settings when more information about the distribution such as upper and lower bounds on the mean are known \cite{L18,H18}. 

Amplitude estimation also has applications which are not reducible to approximate counting, where the goal is to estimate $\bra{0} U\ket{0}$
for a unitary $U$ without additional structure. Some examples of this kind include applications of amplitude estimation to quantum linear algebra and machine learning. For example, quantum procedures for estimating the inner product $\braket{x}{y}$ between vectors $x, y \in \R^{n}$ have been found to be useful for quantum classification and clustering algorithms \cite{KLLP18, WKS14}. Amplitude estimation is the source of the quantum speedup for inner product estimation, with the quantum algorithms requiring $O(1/\epsilon)$ samples as opposed to the $O(1/\epsilon^{2})$ samples required classically for estimating the inner product to error $\epsilon$. Amplitude estimation variants are also important for reducing the condition number dependence in the quantum linear system solvers from $O(\kappa^{2})$ to $O(\kappa)$ \cite{HHL09,A12}. 

In recent times, there has been a lot of interest in reducing the resource requirements for amplitude estimation algorithms, moving towards the goal of finding an AE algorithm compatible with noisy intermediate scale quantum (NISQ) devices \cite{P18}. The major limitation for adapting amplitude estimation to nearer term devices is that the circuit $U$ needs to be run sequentially $O(1/\epsilon)$ times resulting in a high circuit depth for the algorithm. 

The classical amplitude estimation algorithm \cite{BHMT02} invokes the controlled quantum circuit $U$ at least $O(1/\epsilon)$ times in 
series followed by a quantum Fourier transform for estimating amplitudes to error $\epsilon$ using the phase estimation algorithm \cite{K95}. This makes applications of amplitude estimation like Monte Carlo methods or approximate counting prohibitive for near term hardware as, even in cases where the oracle itself does not have significant depth, the high number of repetitions in series of the oracle makes the overall depth of the circuits prohibitive for the high noise rates of current NISQ devices. A significant amount of recent work has tried to make amplitude estimation nearer term. The known results on amplitude estimation along with their resource requirements in terms of qubits, depth, and total number of oracle calls (number of times the circuit $U$ is run) are summarized in Table \ref{table1}.

\begin{table} [H] 
\begin{center} 
  \begin{tabular}{|l|c|c|c|} 
   \hline
  &&&\\
   \textbf{Algorithm} & \textbf{Qubits} & \textbf{Depth} & \textbf{Number of calls}  \\
      &&&\\
    \hline
      &&&\\
      Amplitude estimation \cite{BHMT02} & $n + \log(1/\epsilon)$ & $d \cdot \frac{1}{\epsilon}+\log\log(1/\epsilon)$ & $\frac{1}{\epsilon}$ \\
    &&&\\
    \hline
     &&&\\
    QFT free amplitude estimation \cite{S20, AR20}  & $n$ & $d \cdot \frac{1}{\epsilon}$ & $\frac{1}{\epsilon}$ \\
      &&&\\
     \hline 
     &&&\\
     IQAE \cite{ grinko2019iterative} & $n$ & $d \cdot \frac{1}{\epsilon}  $ & $\frac{1}{\epsilon} \log (1/\epsilon)$  \\
    &&&\\
    \hline 
     &&&\\
    Power-law AE [This paper]  & $n$ &  $d \cdot \left(\frac{1}{\epsilon}\right)^{1-\beta}  $ & $ \left( \frac{1}{\epsilon} \right)^{1+\beta}$  \\
    &&&\\
     \hline 
     &&&\\
      QoPrime AE [This paper]   & $n$ &  $d \cdot \left(\frac{1}{\epsilon}\right)^{1-q/k}$ & $ \left( \frac{1}{\epsilon} \right)^{1+q/k}$  \\
    &&&\\
    \hline  
   \end{tabular} 
   \caption {Asymptotic tradeoffs of amplitude estimation algorithms. Parameters:  $n$ is the number of qubits and $d$ is the circuit depth for a single application of $U$, $\epsilon$ is the additive error, $\beta \in (0,1]$, $k\geq 2$, $q \in [k-1]$. 
   } \label{table1} 
  \end{center} 
   \end{table}

A number of recent works \cite{S20,AR20,grinko2019iterative} have given QFT free amplitude estimation algorithms, that is algorithms that do not require a Quantum Fourier transform (QFT) at the end of the computation. The QFT circuit is applied to a register with $O(\log(1/\epsilon))$ qubits and adds an asymptotic factor of $O( \log \log (1/\epsilon))$ to the overall circuit depth of the algorithm. Eliminating the QFT does not significantly lower the depth of the AE algorithm, but it is an important step towards making amplitude estimation nearer term as it removes the need to apply controlled versions of the oracle. Controlled oracles incur considerable overhead as they require adding multiple controls to every gate in the oracle circuit. The iterative quantum amplitude estimation (IQAE) algorithm in Table \ref{table1} has the same asymptotic performance as \cite{S20, AR20}, however it can be viewed as the state of the art AE algorithm in practice, as it has a provable analysis with the lowest constant overheads among the known QFT free amplitude estimation methods. 

In this work, we take a different view and we ask the following question: can quantum amplitude estimation algorithms that use quantum circuits with depth asymptotically less than $O(1/\epsilon)$ provide any speed up with respect to classical algorithms? We respond to this question by focusing on algorithms that interpolate between classical and quantum amplitude estimation. At a high level, classical amplitude estimation requires $O(1/\epsilon^2)$ applications of the oracle, but these applications can be performed in parallel. The standard amplitude estimation algorithm \cite{BHMT02} on the other hand requires $O(1/\epsilon)$ serial applications of the oracle, but one only needs to perform this computation once. In this paper, we present two different algorithms that interpolate between the classical and quantum settings with an optimal tradeoff  $ND=O(1/\epsilon^{2})$, where $N$ is the total number of oracle calls and $D$ is the maximum number of sequential oracle calls. 

That said, it is known that there are limitations to depth reduction for amplitude estimation and that the circuit depth for AE cannot be reduced generically while maintaining the entire speedup. Zalka first established a tradeoff between the depth and the number of executions for Grover's search \cite{Z99}. Recent work of Burchard \cite{B19}, building upon \cite{J17} extends the tradeoffs in Zalka's work to the setting of approximate counting and shows that for approximate counting, depth-$D$ parallel runs of the algorithm can at most achieve a $D$-fold speedup over the classical algorithm. Burchard's work \cite{B19} also suggests an algorithm matching his lower bound in settings where the approximate counting domain can be partitioned into equally sized pseudorandom subdomains. However, this method is restricted to approximate counting as it assumes a discrete domain that can be partitioned into pseudorandom parts and it incurs a number of overheads that are significant for near term devices, we refer to \cite{bouland2020prospects} for a more detailed discussion. The amplitude estimation algorithms that we propose in this paper match the Burchard-Zalka lower bounds exactly, and are applicable to the more general setting of amplitude estimation where no additional structure required. 

Tradeoffs between depth and the number of oracle calls for quantum algorithms have also been considered in some recent works.  Variational quantum eigensolver (VQE) algorithms  
interpolating between classical sampling and phase estimation with circuit depth $O(1/\epsilon^{\alpha})$ and $O(1/\epsilon^{2(1-\alpha)})$ oracle calls were proposed in \cite{wang2019accelerated}.  
Fisher information calculations similar to the ones used here have been to analyze engineered quantum likelihood functions \cite{koh2020framework} and to to design experiments to 
best estimate the period for time invariant single qubit Hamiltonian systems \cite{ferrie2013best}.

\subsection*{The Power law Amplitude Estimation algorithm} 

Our first algorithm utilizes the framework proposed by Suzuki et al. \cite{S20} for QFT free amplitude estimation. 
A higher-level description of this algorithm, under the name of the Kerenidis-Prakash approach, appears in the recent survey paper \cite{bouland2020prospects}. In the Suzuki et al. framework the oracle is invoked with varying depths, and then measurements in the standard basis are performed, followed by classical maximum likelihood post-processing to estimate the amplitude. Algorithms in this framework are specified as schedules $(m_{k}, N_{k})$ where the oracle is applied $m_{k}$ times in series for $N_{k}$ iterations and, at the end, the results are post-processed classically using maximum likelihood estimation. 

The main idea behind these algorithms is the following. The classical amplitude estimation procedure uses $O(1/\epsilon^{2})$ calls to the circuit $U$ and measurements in the standard basis, which is equivalent to sampling from a Bernoulli random variable with success probability $\alpha = \cos^{2}(\theta)$. The quantum amplitude estimation algorithms on the other hand, use quantum circuits that sequentially perform oracle calls at all depths up to $O(1/\epsilon)$. If the quantum circuits have depth $k$ then a quantum algorithm samples from a Bernoulli random variable with success probability $\cos^{2}((2k+1)\theta)$. Suzuki et al \cite{S20} observed that samples from a Bernoulli random variable are more informative for estimating $\theta$ if the success probability is $\cos^{2}((2k+1)\theta)$, and this can be made precise using the notion of Fisher information $I_f(\alpha)$ for a schedule. 

The QFT free amplitude estimation algorithm \cite{S20} uses an exponential schedule with depths $m_{k}= 2^{k}$ all the way up to the maximum depth of $O(1/\epsilon)$ and chooses $N_{k} = N_{shot}$ to be a constant. The quantum Fourier transform step at the end of the algorithm is eliminated, however the asymptotic efficiency of max-likelihood post-processing is not established rigorously \cite{AR20} . Linear schedules with $m_{k} =k$ were also considered in \cite{S20}. Subsequently, Aaronson and Rall \cite{AR20} provided a provable QFT free amplitude estimation algorithm that does not require max-likelihood estimation. The state of the art QFT free algorithm is the IQAE \cite{grinko2019iterative}, which improves upon the large constant factors required for the analysis in \cite{AR20} and has the best performance among all known QFT free variants of amplitude estimation. 

Our power law amplitude estimation algorithm (Power law AE) uses power law schedules with constant $N_{k} = N_{shot}$  and $m_{k} = \lfloor k^{\frac{1-\beta}{ 2\beta} } \rfloor$, where $k$ starts from $1$ and 
increases until the maximum depth for the quantum circuit is $O(1/\epsilon^{1-\beta})$ for $\beta \in (0,1]$, at a cost of more parallel runs with total number of oracle calls scaling as $O(1/\epsilon^{1+\beta})$. When $\beta$ tends to 0, the schedule approaches the exponential schedule, while when $\beta$ is equal to 1, the algorithm is the classical one. The analysis of the power law AE is based on the observation that maximum likelihood estimation in this setting is equivalent to sub-dividing the domain for the amplitude $\theta$ into $O(1/\epsilon)$ equal parts and performing Bayesian updates starting from a uniform prior. If the prior and the log-likelihood function are sufficiently regular, this allows us to use the Bernstein Von-Mises theorem \cite{HM73}, which can be viewed as a Bayesian central limit theorem that quantifies the rate of convergence to the Bayesian estimator to the normal distribution centered at the true value of $\theta$ with variance $1/N_{shot} I_f(\alpha)$, where $\alpha  = \cos^2(\theta)$. The variant of the Bernstein Von-Mises theorem proved in \cite{HM73} is particularly helpful for the analysis as it bounds the rate of convergence of the posterior distribution to the normal distribution in the $\ell_{1}$ norm. The tradeoff $ND=O(1/\epsilon^{2})$ follows from Fisher information calculations for the power law schedules.

Very recently, super-linear polynomial schedules  in the presence of depolarizing noise have been considered by Tanaka et al. \cite{T20}. We also study the behaviour of our algorithm in the presence of depolarizing noise, and describe a way to make our algorithm robust to noise. In fact, we give a simple method for choosing the optimal parameter $\beta$ given a desired accuracy and noise level. One can see our algorithm as an optimal way of utilizing all the power of the available quantum circuit in terms of depth, meaning that instead of having to wait for quantum circuits to have good enough fidelity to apply sequentially a number of oracle calls of the order of $1/\epsilon$, which for Monte Carlo applications can grow between $10^3$ to $10^6$, our power-law AE algorithm makes it possible to use quantum circuits of any depth and provide a corresponding smaller speedup. 
  
The theoretical analysis described above relies on strong regularity conditions on the log-likelihood and the prior required for the Bernstein Von-Mises theorem, which can be hard to verify rigorously, even though they seem to hold empirically for the log-likelihood function for amplitude estimation.
It is therefore desirable, at least from a theoretical point of view,  to have an AE algorithm achieving the same $ND=O(1/\epsilon^{2})$ tradeoffs that does not rely on these conditions. If we attempt to find a schedule maximizing the Fisher information for a given number of oracle calls, the optimal solution is a schedule that makes oracle calls at the maximal possible depth. However, making oracle calls at a single depth are not sufficient for amplitude estimation due to the periodicity of the function $\cos^{2}((2k+1)\theta)$. This led us to consider AE algorithms that make queries at two (or more) depths and combine the results, leading to the QoPrime AE algorithm. 

\subsection*{The QoPrime Amplitude Estimation algorithm} 

Our second amplitude estimation algorithm (QoPrime) uses a number theoretic approach to amplitude estimation that enables a fully rigorous correctness proof and has the same depth vs number of oracle calls tradeoff as the power-law AE for a discrete set of exponents.

The basic idea for the QoPrime algorithm is to choose $k$ different co-prime moduli, each close to $O(1/\epsilon^{1/k})$ so their product is $N = O(1/\epsilon)$. Let the true value of the amplitude be $\pi M/2N$ where $M \in [0,N]$.  The algorithm estimates $\lfloor M \rfloor \mod N_{i}$, where $N_{i}$ is the product of $q$ out of the $k$ moduli using $N/N_{i}$ sequential calls to the oracle followed by measurements in the standard basis. These low accuracy estimates are then combined using the Chinese remainder theorem to obtain $\lfloor M \rfloor \mod N$. We now sketch the main idea for the QoPrime algorithm for the simplest case of $k=2$ moduli, this 
corresponds to an algorithm with $D=O(1/\epsilon^{1/2})$ and $N= O(1/\epsilon^{3/2})$. Let $a$ and $b$ be the two largest co-prime numbers upper bounded by the maximum depth $D=O(1/\sqrt{\epsilon})$. For simplicity, let us assume that the true value for $\theta$ is of the form $\frac{\pi M}{ 2 (2a+1)(2b+1)}$ for some integer $M \in [ (2a+1)(2b+1)]$. The QoPrime algorithm recovers $M \mod (2a+1)$ and $M \mod (2b+1)$ and then determines $M$ using the Chinese remainder theorem. 

Invoking the oracle $a$ times in series followed by a measurement in the standard basis is equivalent to sampling from a Bernoulli random variable with the success probability $\cos^{2} ((2a+1)\theta)$. As a function of $\theta$, this probability is periodic with period $\frac{\pi}{2a+1}$ and is therefore a function of $\pm M \mod (2b+1)$. Subdividing the interval $[0, \pi/2]$ into $(2b+1)$ equal parts, it follows from the additive Chernoff bounds that $\pm M \mod (2b+1)$ can be recovered with $O((2b+1)^{2})$ samples. The algorithm is able to estimate $M \mod (2b+1)$ with $O(2b+1)^{2}$ repetitions of a quantum circuit of depth $(2a+1)$ and, similarly, $M \mod (2a+1)$ with $O(2a+1)^{2}$ repetitions of a quantum circuit of depth $(2b+1)$. The Chinese remainder theorem is then used to combine these low precision estimates to obtain the integer  $M \in [ (2a+1)(2b+1)]$, thus boosting the precision for the estimation procedure. The total number of oracle calls made was $O( ab^{2} + ba^{2} ) = O( 1/\epsilon^{1.5})$. The maximum depth for the oracle call was $O(1/\epsilon^{1/2})$. A recursive application of the algorithm at a lower precision is needed to determine the sign of the estimates, the details of which are described in section \ref{QoPrime}. The algorithm extends to the more general case where the true value for $\theta$ is $\frac{\pi M}{ 2 N}$ where $N$ is the product of  $q \leq k$ coprime moduli and $M \in [0, N]$, in which case we also show how to pick these values $q,k$. Here, the maximum depth of the quantum circuit is $D=O(1/\epsilon^{1-q/k})$, and the total number of oracle calls are  $ N = O(1/\epsilon^{1+q/k})$.

Further, we study the accuracy of the algorithm in the presence of depolarizing noise and provide a number of graphs that show the behavior of the algorithm under noise. The analysis of the algorithm in a noisy setting shows that noise limits the depth of the oracle calls that can be made, but it also 
allows us to choose optimally the algorithm parameters to minimize the number of oracle calls for a given target approximation error and noise rate. The experiments show that the constant overhead for the QoPrime algorithm is reasonable ($C <10$) for most settings of interest.
\vspace{6pt}

Last, we benchmark our two new low depth AE algorithms with the state of the art IQAE algorithm \cite{grinko2019iterative} in noisy settings. Algorithms such as IQAE require access to a full circuit depth of $O(1/\epsilon)$, and this large depth is exponentially penalized by the depolarizing noise by requiring an exponentially large number of samples to achieve a precision below the noise level. In comparison, the power law and the QoPrime AE algorithms transition smoothly to a classical estimation scaling and do not suffer from an exponential growth in oracle calls. The Power law AE algorithm shows the best practical performance according to the simulations for different error rates and noise levels.

Overall, we present here two low depth algorithms for quantum amplitude estimation, thus potentially bringing a number of applications closer to the NISQ era. Of course, it is important to remember that even applying the oracle $U$ once may already necessitate better quality quantum computers than the ones we have today so observing these quantum speedups in practice is still a long way ahead. Nevertheless, we believe the optimal tradeoff between the total number of oracle calls and the depth of the quantum circuit that is offered by our algorithms can be a powerful tool towards finding quantum applications for near and intermediate term devices. 

The paper is organised as follows: In Section \ref{PLAE}, we describe and analyze the power law amplitude estimation algorithm, while in Section \ref{QoPrime}, we describe and analyze the QoPrime amplitude estimation algorithm. In Section \ref{empirics}, we present empirical evidence of the performance of our algorithm and benchmarks with other state-of-the-art algorithms for amplitude estimation.

\section{Amplitude estimation with power law schedules} \label{PLAE}

\subsection{Preliminaries} 

In this section, we introduce some preliminaries for the analysis of amplitude estimation with power law schedules. Let $X$ be a random variable with density function $f(X, \alpha)$ that is determined by the single unknown parameter $\alpha$. Let $l(X, \alpha) = \log f(X, \alpha)$ be the log-density function and let $l'(x, \alpha) = \part{l(x, \alpha)}{ \alpha}$. In this section, all expectations are with respect to the density function $f(X, \alpha)$ and $'$ denotes the partial derivative with respect to $\alpha$. 

The Fisher information $I_{f} (\alpha)$ is defined as the variance of the log-likelihood function, that is $I_{f} (\alpha)= Var[  l'(X, \alpha)]$. It can also be equivalently defined as $I_{f} (\alpha)=- E_{f} [ l''(X, \alpha) ]$.  
More generally, for parameters $\alpha \in \R^{n}$, the Fisher information is defined as the covariance matrix of the second partial derivatives of $l(X, \alpha)$ with respect to $\alpha$. 

Let $\alpha^{*}$ be the true value for $\alpha$ and consider a Bayesian setting where a prior distribution on $\alpha$ is updated given i.i.d. samples $X_{i}, i \in [n]$ from a distribution $f(X, \alpha^{*})$. 
The Bernstein-Von Mises theorem stated below quantifies the rate of convergence of the Bayesian estimator to the normal distribution with mean and variance $(\alpha^{*}, \frac{1}{nI_{f}(\alpha^{*})})$
in the $\ell_{1}$ norm for cases where the log-likelihood and the prior are sufficiently regular. The complete list of regularity conditions for the theorem is given in Appendix \ref{BvM}. 

\begin{theorem} \label{thm:bvm} 
[Bernstein Von-Mises Theorem \cite{HM73} ]  Let $X_{i}, i \in [n]$ be independent samples from a distribution $f(X, \alpha^{*})$ and let $R_{0}$ be the prior distribution on $\alpha$. 
Let $R_{n}$ be the posterior distribution after $n$ samples and let $Q_{n, \alpha^{*}}$ be the Gaussian with mean and variance $(\alpha^{*}, \frac{1}{nI_{f}(\alpha^{*})})$.

If  $f(X, \alpha^{*})$ and $R_{0}$ satisfy the regularity conditions enumerated in the Appendix \ref{BvM}, then there exists a constant $c>0$ such that, 
\eq{ 
\Pr_{X_{i} \sim f, i \in [n]} \left[ \norm{ R_{n} - Q_{n, \alpha^{*}} }_{1} \geq c \sqrt{\frac{ 1 }{ n}}  \right ]  = o\en{\frac{1}{\sqrt{n}}} 
} 
\end{theorem}

As we have defined before, the amplitude estimation algorithm has access to a quantum circuit $U$ such that 
$U \ket{0^{t} } = \cos(\theta) \ket{x, 0} + \sin(\theta) \ket{x', 0^{\perp}}$ where $\ket{x}, \ket{x'}$ are arbitrary states on $(t-1)$ qubits. 
If the second register is measured in the standard basis, then the distribution of the measurement outcome $f(X,\alpha)$ is a Bernoulli random variable with success probability $\alpha = \cos^2(\theta) \in [0,1]$. If a quantum circuit of $k$ sequential calls to the circuit $U$ is applied then the quantum 
state $\cos ((2k+1) \theta) \ket{x, 0} + \sin((2k+1)\theta ) \ket{x', 0^{\perp}}$ can be obtained. Measuring this state in the standard basis, the distribution of measurement outcome is again a Bernoulli random variable with success 
probability $\cos^{2} ((2k+1)\theta)$. 

A quantum AE algorithm therefore has access to samples from Bernoulli random variables with success probability $\cos^{2} ((2k+1)\theta)$ where $k$ is the number of sequential oracle calls in the quantum circuit, which corresponds to its depth. The higher depth samples are more informative for estimating 
$\theta$. The next proposition quantifies the advantage for higher depth samples, showing that the Fisher information grows quadratically with the depth of the oracle calls. 

\begin{proposition} \label{propf} 
Let $f(X, \alpha) = \beta^{X} (1- \beta)^{1-X}$ for parameter $\alpha = \cos^{2} (\theta)$ and $\beta = \cos^{2} ((2m_{k}+1) \theta )$ for a positive integer $m_k$. The Fisher information is $I_{f}(\alpha) = 
\frac{ (2m_{k}+1)^{2} } { \alpha (1-\alpha) }$. 
\end{proposition} 
\begin{proof} 
As  $\alpha = \cos^{2} (\theta)$ we have $d\alpha = 2 \cos(\theta) \sin(\theta) d\theta$. The log-likelihood function is $l(X, \alpha) = X \log \beta + (1-X) \log (1-\beta)$. Thus, 
\al{ 
I_{f}(\alpha) &= -E_{f} \left [ \frac{d}{d\alpha^{2}} (2X \log \cos((2m_{k}+1) \theta)  + 2(1-X) \log \sin((2m_{k}+1) \theta) ) \right ] \nl 
&= \frac{-1}{ 2 \alpha (1-\alpha) } E_{f} \left [ \frac{d}{d\theta^{2}} (X \log \cos((2m_{k}+1) \theta)  + (1-X) \log \sin((2m_{k}+1) \theta) ) \right ] \nl 
&= \frac{(2m_{k}+1)^{2}}{ 2 \alpha (1-\alpha) } \left (  \frac{ E_{f}[X] } { \cos^{2} ((2m_{k}+1) \theta) }  + \frac{E_{f} [1-X]}{ \sin^{2} ((2m_{k}+1) \theta)}  \right ) \nl 
&=  \frac{(2m_{k}+1)^{2}}{  \alpha (1-\alpha) } \notag . 
} 

\end{proof}

\noindent The Fisher information is defined as a variance and is therefore additive over independent samples that do not need to be identically distributed. The Fisher information of an amplitude estimation schedule $(m_{k}, N_{k})$ 
\cite{S20} is the sum of the Fisher informations for the individual samples.

\subsection{The Power law Amplitude Estimation algorithm} 

The amplitude estimation algorithm using power law schedules is given as Algorithm \ref{AE:pls}. It is then analyzed to establish the tradeoff between the depth and the total number of oracle calls in a setting where the Bernstein Von Mises Theorem is applicable. 

\begin{algorithm} [H]
\caption{The Power law Amplitude Estimation} \label{AE:pls} 
\begin{algorithmic}[1]
\REQUIRE Parameter $\beta \in (0,1]$, $N_{shot} \in \Z$ and desired accuracy $\epsilon$ for estimating $\theta$. 
\REQUIRE Access to a unitary $U$ such that $U\ket{0} = \cos (\theta) \ket{x,0} + \sin(\theta) \ket{x',1}$.
\ENSURE An estimate of $\theta$ within accuracy $\epsilon$ with high probability
\STATE  Initialize the prior to be the uniform distribution on angles $\theta = \frac{\pi t \epsilon}{2}$.  
\FOR{k=1 \TO $K=\max \en{ \frac{1}{\epsilon^{2\beta}} , \log (1/\epsilon) } $ }
\STATE Initialize $N_{k_0}=0$ and  $N_{k_1}=0$. 
\FOR{i=1 \TO $N_{shots}$}
\STATE Apply $m_{k} = \lfloor k^{\frac{1-\beta}{2\beta} } \rfloor$ sequential oracle calls and measure last qubit of resulting quantum state in the standard basis.
\STATE If the outcome is 0, then $N_{k_0}=N_{k_0}+1$, else $N_{k_1}=N_{k_1}+1$. \ENDFOR
\STATE Perform Bayesian updates $p(\theta) \to p(\theta) \cos((2m_{k}+1) \theta)^{N_{k_0}}  \sin((2m_{k}+1) \theta)^{N_{k_1}} $ for $\theta=\pi t \epsilon/2$ for integer valued $t\in [0, 1/\epsilon]$ and interpolate to obtain the posterior probability distribution. \ENDFOR
\STATE Output $\theta$ with the highest probability according to the posterior probability distribution.

\end{algorithmic}
\end{algorithm}

\noindent The next theorem shows that Algorithm \ref{AE:pls} achieves approximation error $\epsilon$ with parameters $N= O( \frac{1}{ \epsilon^{1+\beta}})$ and $D= O( \frac{1}{ \epsilon^{1-\beta}})$ 
where $D$ is the maximum depth of the oracle calls and $N$ is the total number of oracle calls made. The choice of $K=\max (\frac{1}{\epsilon^{2\beta}}, \log (1/\epsilon))$ ensures that our power law AE algorithm makes a sufficient number of queries for small $\beta$ and approaches the exponential schedule of \cite{S20} as $\beta \to 0$. 

\begin{theorem} \label{thm:Bvm_proof_of_MLE}
The Power law Amplitude Estimation algorithm \ref{AE:pls} outputs an $\epsilon$ accurate estimate with $N= O( \frac{1}{ \epsilon^{1+\beta}})$ oracle calls and maximum depth $D= O( \frac{1}{ \epsilon^{1-\beta}})$ with probability at least $0.9$, that is the algorithm attains the tradeoff $ND= O(\frac{1}{\epsilon^{2}})$ in settings where the Bernstein-Von Mises theorem is applicable. 
\end{theorem} 

\begin{proof} 
The total number of oracle calls for Algorithm \ref{AE:pls} is  $N= \sum_{k \in [K] } N_{shot} (2m_{k}+1)$ while the Fisher information for the power law schedule can be computed as $I_{f}(\alpha)= \frac{N_{shot}}{\alpha(1-\alpha)}  \sum_{k \in [K]}  (2m_{k} + 1)^{2}$ using proposition \ref{propf}.  Approximating the sums in $N$ and $I_{f}(\alpha)$ by the corresponding integrals we have $I_{f}(\alpha) = O(K^{1/\beta})$,  $N= O(K^{(1+\beta)/2\beta})$
and maximum depth $D=O(K^{(1-\beta)/2\beta})$, so that $ I_{f}(\alpha) = O(ND)$. Note that for $K=\max \en{ \frac{1}{\epsilon^{2\beta}} , \log (1/\epsilon)}$, we have 
$N= O( \frac{1}{ \epsilon^{1+\beta}})$ and $D= O( \frac{1}{ \epsilon^{1-\beta}})$ and $I_{f}(\alpha) = O( \frac{1}{\epsilon^{2}})$.

It remains to show that with probability at least $0.9$, the estimate output by the algorithm is within additive error $\epsilon$ of the true value.  
Applying the Bernstein Von-Mises Theorem, the $\ell_{1}$ distance between the posterior distribution and the Gaussian with mean and variance $(\alpha^{*},  1/ \sqrt{N_{shot}I_{f}(\alpha^*)})$
is at most $c/\sqrt{N_{shot}}$ with probability at least $1- c'/\sqrt{N_{shot}}$ for some constants $c, c'>0$. 

Choosing $N_{shot}> \en{ \frac{\max(c, c')}{\delta}} ^{2}$ for $\delta=0.05$, the estimate is within $(\alpha^{*} \pm \frac{3 \delta}{c \sqrt{ I_{f}(\alpha^{*}) }})$ with probability at least $1-\delta (1+1) - 0.0013 \geq 0.89$. The success probability can be boosted to $1-\frac{1}{poly(\zeta)}$ for $\zeta>0$ by running the algorithm $O(\log (1/\zeta))$ times and outputting the most frequent estimate. 

\end{proof} 

\noindent The above proof analyzes Algorithm  \ref{AE:pls}  in settings where the Bernstein Von Mises theorem is applicable. The complete list of regularity conditions required for the theorem are enumerated in Appendix \ref{BvM}. In high level, for some neighborhood around the real value of $\theta$ they impose:  the smoothness of the prior distribution; the smoothness of the density function $f(X, \theta)$, which will be satisfied if the norm of log-likelihood is bounded around $\theta$; and the smoothness and differentiability of the Fisher information around $\theta$, which will be true if the log-likelihood function has derivatives of order at least $3$. 

Figure 1 plots the log-likelihood function for the power law AE for a fixed exponent $\beta$ and varying depths and a true value for $\theta$ chosen 
uniformly at random from $[0, \pi/2]$. The figure illustrates that the log-likelihood function is smooth in a neighborhood around the true value indicating that the regularity conditions for the Bernstein-Von Mises theorem are plausible in this setting. Adding noise may further regularize the log-likelihood functions and enforce the regularity conditions required for the Bernstein-Von Mises theorem. 

The constants $c, c'$ in the proof of Theorem \label{thm:Bvm_proof_of_MLE} can be determined explicitly if the regularity conditions are verified giving an explicit value for $N_{shot}$. In the absence of an explicit value, experimental results demonstrating convergence for a small value of $N_{shot}$ can be taken as evidence that $N_{shot}$ is a moderately small constant.


\subsection{Power law amplitude estimation with noise} 

We perform here an analysis similar to the work in \cite{T20}. Noise provides a natural constraint on accessible circuit depths and noise models exponentially penalize larger depths, leading to an exponential decoherence of quantum states; therefore, exponentially more classical samples are needed to battle the noisy information. In this section, we explain this effect on our algorithm in the case of the depolarizing noise model. 

\begin{proposition}[\cite{T20}]\label{prop:prob_under_noise}
Assuming a depolarizing noise channel with a per-oracle-call rate of $\gamma \geq 0$, if measurements in the standard basis are performed on a quantum circuit of $k$ sequential calls to the oracle $U$, the distribution of measurement outcomes is a Bernoulli random variable with success probability 
\begin{equation}\label{eq:depoldef}
 p=\frac{1}{2}- \frac{1}{2}e^{-\gamma (2k+1)}\cos(2 (2k+1)\theta).
 \end{equation}
\end{proposition}

Let $\{m_k\}_{k=0,\dots, K}$ be a schedule. The Fisher information with respect to the angle $\theta$ in the presence of depolarizing noise with parameter $\gamma$ is given by 

\begin{align}\label{fisher_noisy} 
	I_f(\theta)=4N_{shot}\sum_{k=0}^K (2m_k+1)^2\frac{e^{-2\gamma (2m_k+1)} \sin^2(2(2m_k+1)\theta)}{1-e^{-2\gamma (2m_k+1)} \cos^2(2(2m_k+1)\theta)},
\end{align} 

\noindent
see \cite{T20} for a proof. Recall that $\alpha = \sin^2(\theta)$ is the probability of success without depolarizing noise. For $k$ such that $\gamma m_k$ becomes bigger than some large enough constant, the second term in the sum will be exponentially suppressed, the Fisher information will not increase significantly even if we keep increasing our depth. In practice, we do not want to use a Fisher information which is dependent on the parameter $\theta$ to be estimated; we can obtain a simple $\theta$-independent upper bound to \eqref{fisher_noisy} by using the inequalities  $\frac{1}{1-x}\leq 1+x$ and $e^{x} \geq 1+x$, which gives us:
\begin{equation}\label{eq:fisher_approx}
	I_f(\theta) \leq 4N_\text{shot}\sum_{k=0}^K(2m_k+1)^2e^{-2\gamma m_k}\,.
\end{equation}

Let us consider the power law schedule given by $m_k=\lfloor  k^\eta \rfloor$ for $k=0,1,\dots, K$ for $\eta= \frac{1-\beta}{2\beta}$ and $\beta \in (0,1]$, as defined in Algorithm \ref{AE:pls}. We see that 
\begin{align*}
	I_f(\theta) \leq &\sum_{\gamma \lfloor  k^\eta \rfloor\leq 1}(2\lfloor  k^\eta \rfloor+1)^2C+\text{ exponentially suppressed terms }
	\\ & \leq 10C\sum_{k=0}^{(1/\gamma)^{1/\eta}}(2k^\eta+1)^2 \leq  C'/\gamma^{2+1/\eta} = C'/\gamma^{2/(1-\beta)}.
\end{align*}
Here $C$ is a constant that can taken to be $C:=\frac{4N_{shot}}{\sin^4(2\theta)}$ (if this diverges, add small random noise to $\theta$).
So by the Cram\'er-Rao bound, we have that 
\begin{equation}\label{cramer}
	I_f(\theta) \leq \frac{C'}{\gamma^{2/(1-\beta)}}\Rightarrow \epsilon \geq c\gamma^{1/(1-\beta)}, 
\end{equation}
where $c:=C'^{-1/2}$.
This implies that given a noise level $\gamma$ we cannot get an error rate $\epsilon$ smaller than $c\gamma^{1/(1-\beta)}$ by increasing the depth for a power law schedule with a \emph{fixed} parameter $\beta$. Seeing this tradeoff from the other side, given a desired error rate $\epsilon$ and a noise level $\gamma$, we know how to pick the parameter $\beta$ to match the lower bound in equation (\ref{cramer}). Here we are assuming that the regularity conditions of Bernstein-von Mises also holds in the noisy setting, so that we are able to match the lower bound on $\epsilon$. 

If the noise level is smaller than the desired error rate, then we can pick the exponential schedule (for $\beta$ equal to 0) since we can apply circuits of depth up to $O(1/\epsilon)$. For $\epsilon<\gamma$, we assume that we can match the lower bound in equation (\ref{cramer}) so that $\beta$ can be picked as follows.

\begin{proposition}\label{prop:pwl_exponent_noise}
	Assume we are given as input the target error $\epsilon$ and noise level $\gamma$ with $0 <  \epsilon < \gamma < 1$ and that the regularity conditions of the Bernstein-von Mises hold in the noisy setting. The parameter $\beta$ of the power law algorithm can be picked as
	\[ \beta \geq 1 - \frac{ \log {\gamma}}{\log {\epsilon/c} }
	\]
	to achieve $\epsilon\leq c\gamma^{1/(1-\beta)}$ for the powerlaw schedule with parameter $\eta = \frac{1-\beta}{2\beta}$ (with high probability).
\end{proposition}


%

\section{QoPrime: A number theoretic amplitude estimation algorithm} \label{QoPrime}

\subsection{Preliminaries} 
We introduce the main technical tools needed for the QoPrime AE algorithm, these are the Chinese remainder theorem and the additive form of the Chernoff bounds. 
\begin{theorem} \label{crt} 
[Chinese Remainder Theorem] Let $a_{i} \in \N, i \in [k]$ be pairwise coprime numbers such that for all $i\neq j$, $gcd(a_{i}, a_{j})=1$ and let $N=\prod_{i \in [k]} a_{i}$. 
Then for all $b_{i} \in [a_{i}], i \in [k]$ there is an efficient algorithm to find $M \in [N]$ such that $M \mod a_{i} = b_{i}$. 
\end{theorem} 
\begin{proof} 
First we provide the proof for $k=2$. Applying the extended Euclidean algorithm we can find $u_{1}, u_{2} \in \Z$ such that, 
\al{ 
1= gcd(a_{1}, a_{2}) = u_{1} a_{1} + u_{2} a_{2} 
} 
Then $M= (b_{2} u_{1} a_{1} + b_{1} u_{2} a_{2}) \mod a_{1} a_{2} $ satisfies $M = b_{i} \mod a_{i}$ for $i = 1,2$. 

For the proof of the general case, note that the $k-1$ numbers 
$a_{1}a_{2}, a_{3}, a_{4}, \cdots, a_{k}$ are coprime and by the above argument the constraints $M= b_{i} \mod a_{i}$ for $i=1,2$ is equivalent to $M = (b_{2} u_{1} a_{1} + b_{1} u_{2} a_{2}) \mod a_{1} a_{2} $. 
The procedure can therefore be repeated iteratively to find the desired $M$.
 
\end{proof} 
\noindent In this paper, we will be using relatively coprime moduli $a_{i}$, however the results can easily be adapted to a setting where the $a_{i}$ are not 
pairwise coprime replacing $N=\prod_{i \in [k]} a_{i}$ by the least common multiple of the $a_{i}$. We define explicitly the bijection given by the Chinese remainder theorem 
as it is used later in the algorithm. 

\begin{defn} \label{def:crt} 
Given pairwise co-prime moduli $N_{i}, i \in [k]$, let $N= \prod_{i \in [k]} N_{i}$. The function $CRT: \prod \Z_{N_{i}} \to Z_{N}$ on input $( M_{1}, M_{2}, \cdots, M_{k})$ evaluates to the unique $M \in [N]$ given by Theorem \ref{crt} such that $M= M_{i} \mod N_{i}$. 
\end{defn}

The second tool needed for the QoPrime algorithm is the additive form of the Chernoff bound and some supplementary calculations on the entropy of the binomial distribution. 
\begin{theorem} [Chernoff-Hoeffding Bound \cite{H94}] \label{chb} 
Let $X_{i}$ for $i \in [m]$ be i.i.d. random variables such that $X_{i} \in \{0,1 \}$ with expectation $E[X_{i}]=p$ and let $\epsilon>0$ and $X=\frac{1}{m}  \sum_{i \in [m]} X_{i}$. Then, 
\begin{enumerate} 
\item $\Pr [ X > p +  \epsilon  ] \leq  e^{-D(p+\epsilon || p) m}$. 
\item $\Pr [ X < p -  \epsilon  ] \leq  e^{-D(p-\epsilon || p) m}$. 
\end{enumerate} 
where the relative entropy $D(x||y) = x \ln \frac{x}{y} + (1-x) \ln \frac{(1-x)}{(1-y)}$ where the relative entropy can be lower bounded using the inequality $D(x|| y) \geq \frac{(x-y)^{2}} {2 \max(x, y) }$ for all $x, y \in [0,1]$.  
\end{theorem} 
The lower bound on the relative entropy is required to make the Chernoff bounds effective, we
and derive a corollary that will be useful for the analysis of the QoPrime algorithm. 
\begin{corollary} \label{cor:chb} 
The following lower bound holds for the relative entropy for all $x, y \in [0,1]$, 
$$D(x|| y) \geq \frac{(x-y)^{2}} {2 \max(1-x, 1-y) }$$ . 
\end{corollary} 
\begin{proof} 
The relative entropy is symmetric under the substitution $(x, y) \to (1-x, 1-y)$, that is $D(x||y) = D((1-x)|| (1-y))$. The result follows by applying the the inequality $D(x|| y) \geq \frac{(x-y)^{2}} {2 \max(x, y) }$
for $(x', y') = (1-x, 1-y)$. 
\end{proof} 
We also state the Multiplicative Chernoff bounds which will be used to compute some constants in the Qo-Prime algorithm. 
\begin{theorem} [Multiplicative Chernoff Bounds]
Let $X_{i}$ for $i \in [m]$ be independent random variables such that $X_{i} \in [0,1]$ and let $X= \sum_{i \in [m]} X_{i}$. Then, 
 $\Pr [ |X - \E[X]| \geq \beta \E[X] ] \leq e^{ -\beta^{2} \E[X]/3}$ for $0< \beta <1$. 
\end{theorem}

\subsection{The QoPrime AE algorithm} 
The QoPrime AE algorithm is presented as Algorithm \ref{AE:crt}. The implementation of the steps of the algorithm is described in greater detail below and the algorithm is analyzed to establish correctness and bound the running time. 

An amplitude estimation algorithm has access to a quantum circuit $U$ such that $U\ket{0^t} = \cos (\theta) \ket{x, 0} + \sin(\theta) \ket{x', 1}$
where $\ket{x}, \ket{x'}$ are arbitrary states on $(t-1)$ qubits. More precisely, let $R_{0}$ 
be the reflection in $\ket{0^{t}}$, that is $R_{0} \ket{0^t} = \ket{0^t}$ and $R_{0} \ket{0^{\perp}  } = - \ket{0^{\perp}}$ and let $S_{0}$ be the reflection on 
$\ket{0}$ in the second register, that is $S_{0} \ket{x, 0} = \ket{x,0}$ and $S_{0} \ket{x,1} = - \ket{x,1}$ for all $\ket{x}$. Like the standard AE algorithm, the QoPrime algorithm uses $(2k+1)$ 
sequential applications of the circuit for $U$ to create the states, 
\all{ 
\ket{\phi_{k} } := (U R_{0} U^{-1} S_{0} )^{k}  U\ket{0}  = \cos ((2k+1)\theta) \ket{x, 0} + \sin((2k+1) \theta) \ket{x', 1}
} {eq:ae} 
An oracle call refers to a single application of the circuit $U$, the total number of oracle calls made is a measure of the running time for an AE algorithm. The maximum circuit depth for an amplitude estimation procedure is the number of sequential calls to $U$. The creation of a single copy of the state $\ket{\phi_{k} }$ in equation \eqref{eq:ae} requires $(2k+1)$ oracle calls.

The QoPrime algorithm is parametrized by integers $(k, q)$ where $k \geq 2$ and $1 \leq q \leq (k-1)$, the parameter $k$ is the number of moduli used for the reconstruction procedure
while $q$ determines the number of moduli that are grouped together. 
The algorithm starts by choosing nearby odd coprime integers $(n_1, \, n_2, \dots, n_k)$, where each coprime is an integer approximately equal to $\left(\frac{\pi}{2\epsilon}\right)^{1/k}$ . For the asymptotic purposes, all we need is that each coprime is therefore $n_i=\Theta(\epsilon^{-1/k})$.
The $k$ coprimes are partitioned into $k/q$ groups, of size at most $q$. Let $\pi_i \subset [k]$ be the subset of coprimes in group $i$, and let $N_{i}= \prod_{j \in  \pi_{i}} n_{j}$ for $i \in \lceil k/q \rceil$ be the product of the 
coprimes in each group. For each group $i$, we will sample at a depth of $N / N_i$.

Let $\theta = \frac{ \pi M}{ 2N}$ be the true value of $\theta$ for $M = \lfloor M \rfloor + \{ M \}$ where $M \in [0,N] $ and $N= \prod_{i \in [k]} n_{i}$. 
The QoPrime algorithm reconstructs estimates $\overline{M_{i}}$ that are within a $1/2$ confidence interval around $M \mod N_{i}$ with high probability. These estimates are 
constructed using samples from depth $N/N_{i}=  O(1/\epsilon^{1-q/k} )$ sequential calls to $U$ to prepare the states in equation \eqref{eq:ae} followed by $O(N_{i}^{2})$ measurements in the standard basis and 
reconstruction using the Chernoff bounds. For each group $\pi_i$ of coprimes, measurements at depth $N/N_i$ will give us information about the moduli $M \mod n_j$, for the coprimes in this particular group $n_j \in \pi_i$.

By repeating this for all groups $\pi_i$, we acquire information about all moduli $M \mod n_j$. These low-precision estimates are then combined using the Chinese remainder theorem to recover $M \mod N$, and therefore an $\epsilon$-accurate estimation for $\theta=  \frac{ \pi M}{ 2N}$. As we shall see, there is a sign ambiguity associated with each group of coprimes, and a recursive call is needed to determine $\theta$
unambiguously. The QoPrime algorithm makes a total $\tilde{O}(1/\epsilon^{1+q/k} )$
oracle calls for estimating $\theta$ within accuracy $\epsilon$ where $\tilde{O}$ hides factor that are logarithmic in $k, q$ and the success probability for the algorithm. It trades off the maximum circuit depth needed for amplitude estimation against the total 
number of oracle calls.

\begin{algorithm} [H]
\caption{The QoPrime Algorithm for Amplitude Estimation} \label{AE:crt} 
\begin{algorithmic}[1]
\REQUIRE Accuracy $\epsilon$ for estimating $\theta$ and parameters $(k, q)$ where $k \geq 2$ is the number of moduli and $1\leq q \leq (k-1)$. 
Desired success probability $p$ and value $c$ such that $1-2ke^{-2c} > p$. 
\REQUIRE Access to unitary $U$ such that $U\ket{0^T} = \cos (\theta) \ket{x, 0} + \sin(\theta) \ket{x', 1}$. 
\ENSURE An estimate of $\theta$ within accuracy $\epsilon$ with probability at least $p$. 

\IF{$q/k> 1/3$}
\STATE Using at most $\frac{24c}{\epsilon^{1+ q/k}}$ samples from $U\ket{0^T}$, find $\theta'$ such that $|\theta'- \theta| \leq \frac{1}{2}\epsilon^{1-q/k}$ with probability at least $1- e^{-2c}$.
\ELSE 
\STATE Invoke the QoPrime algorithm recursively with accuracy $\frac12\epsilon^{1-q/k}$ and parameters $(q',k')=(2q, k)$ such that $1+q/k < 1+ q'/k' < \frac{(1+q/k)} { (1-q/k)} $ to obtain $\theta'$ such that  $|\theta'- \theta| \leq \frac{1}{2}\epsilon^{1-q/k}$.
\ENDIF

\STATE Select $k$ adjacent coprimes from the table in Section \ref{sec33} starting at $\lfloor 1/\epsilon^{1/k} \rfloor$ with product closest to $\pi/(2 \epsilon)$ in absolute value, and let $N=\prod_{i \in [k]} n_{i}$ be their product. \\

\STATE Partition $[k]$ into $K:=\lceil k/q \rceil$ groups $\pi_{i}$ of size at most $q$ and let $N_{i}= \prod_{j \in \pi_{i}} n_{j}$. 

 \STATE Special case: If $|\theta' - \pi/4|  \leq \frac{\min_i N_i }{4N}$, then invoke the QoPrime algorithm with the effective oracle $U'$ using the exact amplitude amplification technique \cite{BHMT02} as described in Lemma \ref{lemma:exactAA}.
 
\FOR{i=1 \TO K}
\STATE Prepare $100c N_{i}^{2}$ copies of $\ket{\phi_{(N-N_{i})/2N_{i}}}$ (see equation \eqref{eq:ae}) and measure in the standard basis. 
\STATE Compute $\widehat{l}_{i}=  \frac{2N_{i}} { \pi} \text{arccos} ( \sqrt{\widehat{p}} )$ where $\widehat{p}$ is the observed probability of outcome $0$.
\ENDFOR

\FOR{all possible sign ambiguity resolutions $s\in\{-1,1\}^{ K } $} 
\STATE For all $j \in [K]$, compute $\overline{M_{j}} = s_{j}\widehat{l}_{j} \mod N_{j}$. 
\STATE For all $j \in [K]$, compute $M_{j} = \lfloor \overline{M_{j}} + \beta_{j}  \rfloor$ where $|\beta_{j}|  \leq 1/4$ are such that the fractional parts $\{ \overline{M_{j}} + \beta_{j} \} = \alpha$ for some $\alpha \in [0,1]$. 
\STATE Let $ \overline{M}  = CRT ( M_{1}, M_{2}, \cdots, M_{K})$, compute the sign-dependent estimate $\theta_{s} = \frac{\pi (\overline{M} + \alpha)  }{2N}$. 
\ENDFOR 

\STATE Output the  choice of angle $\theta_{s}$ that minimizes $|\theta_s  - \theta'|$ over all possible choices of $s\in\{-1,1\}^{ K } $.

\end{algorithmic}
\end{algorithm}

The procedures used in the individual steps of the QoPrime Algorithm \ref{AE:crt} are described next and correctness of the steps is established. Let $ M  = t N_{i} + l$ for some $t \in \Z$ and $0 \leq l \leq N_{i}$. Note that $(-1)^{t}M \mod N_{i} $ is the value being estimated in step 11 of the QoPrime algorithm making it necessary to determine the parity of $t$ (cf. equation \eqref{five}). The approach taken to resolve this ambiguity is to compute all the estimates 
corresponding to the possible parities of $t$ and comparing with an estimate of $\theta$ obtained recursively using the QoPrime algorithm with lower precision in steps 1-5 of the QoPrime algorithm. 
The first lemma in the analysis establishes the correctness of this procedure.

\begin{lemma} \label{lem:zero} 
The recursive procedure of the QoPrime algorithm terminates in at most $O(\log k)$ iterations and outputs an additive error $\min_{i \in [k]} (N_{i}/2N)$ estimate for $\theta$ with probability at least $1- 6e^{-2c}$. 
\end{lemma}  
\begin{proof} 
First we establish the correctness of the stopping condition in step 5. If $q/k\geq 1/3$ then $\frac{(1+q/k)}{2} \geq (1-q/k)$ and thus by the multiplicative Chernoff bounds 
 an additive error $\frac{\epsilon^{(1-q/k)}}{2}$ estimate for $\theta$ is obtained with $\frac{24c}{\epsilon^{(1+q/k)}}$ samples with probability at least $1- e^{-2c}$ for some constant $c>0$. 
 It remains to show that the recursion terminates in at most $O(\log k)$ steps and that the total number of oracle calls used in the recursive step is upper bounded by $O(1/\epsilon^{1+q/k})$. 

The recursive call to the QoPrime algorithm in step 4 uses $O(1/\epsilon'^{(1+q'/k')}) = O(1/\epsilon^{(1-q/k)(1+q'/k')})$ total oracle calls. As $(1-q/k)(1+q'/k')< (1+q/k)$ the total number of oracle calls 
used by steps 4-8 is at most $\tilde{O}(1/\epsilon^{1+q/k})$. The extra oracle calls used for these steps do not change the asymptotic number of oracle calls used by the QoPrime algorithm. Further, as 
$1+q/k < 1+ 2q/k < (1+ q/k)/(1-q/k)$ it is always feasible to choose $q'=2q, k'=k$ and with this choice the stopping condition $q/k\geq 1/3$ in step 5 will hold after at most $O(\log k)$ iterations.

The success probability for these steps is the same as the success probability for the final recursive call to the QoPrime algorithm in step 4, which uses at most $\lceil k/q \rceil \leq 3$ moduli by the stopping condition. 
 Further, Theorem \ref{thm1} shows that the the QoPrime algorithm succeeds with probability at least 
$1- 2ke^{-2c}$, implying that the recursive procedure succeeds with probability at least $1- 6e^{-2c}$.
\end{proof} 
The analysis shows that the QoPrime correctly estimates $\theta$ apart from one exceptional case when the angle is close to $\pi/4$ that is dealt separately in step 8. This case will be discussed later in the analysis
of the sign resolution procedure. 

Step 10 of the QoPrime algorithm computes a quantum state which can be measured to sample from a Bernoulli random variable with success probability $p=\cos^{2} (  \frac{ (t N_{i} + l ) \pi}{2N_{i} })$. 
Step 11 computes an estimate $\widehat{l} = \frac{2N_{i}} { \pi} \text{arccos} ( \sqrt{\widehat{p}} ) $ where $\widehat{p}$ is the observed probability of outcome $0$. 
The analysis below shows that $|\widehat{l} - (-1)^{t} l \mod N_{i}  | \leq 0.25$ with high probability.

In order to analyze these steps, we begin with the observation that if $p = \widehat{p}$ then $\widehat{l} = (-1)^{t} l \mod N_{i}$. 
\all{ 
\frac{2N_{i}} { \pi} \text{arccos} ( \sqrt{p} ) = \begin{cases} 
\frac{2N_{i}} { \pi} \text{arccos} \en{ \cos ( \frac{l\pi}{2N_{i}}) } \text { [if $t= 0 \mod 2$]}   \\
 \frac{2N_{i}} { \pi} \text{arccos} \en{ \sin ( \frac{l\pi}{2N_{i}}) } \text { [if $t= 1 \mod 2$]} 
 \end{cases}  
 =  (-1)^{t} l \mod N_{i} .
 } {five} 
The next Lemma quantifies the error made by the algorithm in approximating $(-1)^{t} l \mod N_{i}$ using the Chernoff bounds to bound the 
difference between $\widehat{p}$ and $p$.

\begin{lemma} \label{lem:one} 
Given integer $N$ and $m=100cN^{2} $ samples,  steps 10-11 of the QoPrime algorithm finds an estimate such that $|\widehat{l} - (-1)^{t} l \mod N | \leq 0.25$ with probability 
at least $1- 2e^{-2c}$. 
\end{lemma} 
\begin{proof} 
Define $F:  [0,1] \to [ 0, N]$ as $F(p)= \frac{2N} { \pi} \text{arccos} ( \sqrt{p} )$ so that $F(p)= (-1)^{t} l \mod N$ by equation \eqref{five} and $F(\widehat{p})= \widehat{l}$. It is sufficient to show that $\Pr[ |F(p) -F(\widehat{p})|  \geq 0.25] \leq e^{-2c}$ for all $p \in [0, 1]$. 

The inverse function $G: [0, N] \to [0,1]$ such that $G(y) = \cos^{2} ( \frac{\pi y}{2N } )$ is monotonically decreasing. Let $F(p)=y$, then $F(p) -F(\widehat{p})  \geq 0.25$ is equivalent to 
$(y-0.25) \geq F(\widehat{p})$, that is $G(y-0.25)< \widehat{p}$. Applying the additive Chernoff bound, 
\al{ 
\Pr [ \widehat{p} > G(y-0.25)   ] &\leq e^{ -m D( G(y- 0.25) || G(y)) } 
} 
The analysis splits into two cases, where the relative entropy is lower bounded using either the inequality in Theorem \ref{chb} or Corollary \ref{cor:chb}. 
Let $A = \frac{\pi y}{2N }$ and $B= \frac{\pi (y-0.25)}{2N }$, first consider the case $y>N/2$ or equivalently $A> \pi/4$, 
\all{ 
m D( G(y- 0.25) || G(y)) &\geq m \frac{ (\cos^{2} ( A ) -  \cos^{2} ( B ))^{2}  } { 2\cos^{2} ( B ) } = m \frac{ (\cos( 2A ) -  \cos ( 2B )^{2}  } { 8\cos^{2} ( B ) } \nl 
&=  m \frac{ \sin^{2} (A+B) \sin^{2} (A-B)  } { 2\cos^{2} ( B ) } \nl 
&= \frac{m \sin^{2} (A-B) }{ 2} (\sin(A) + \cos(A)\tan B))^{2} \nl 
& \geq 50 c N^{2}  \sin^{2} (\frac{\pi}{8N}) \geq 50 c N^{2}  \frac{\pi^{2} }{128N^{2}}  > 2 c. 
} {eight} 
Note that for the last two steps in the computation, the inequality $(\sin(A) + \cos(A)\tan B))^{2}>1$ for $A>\pi/4$ and $A-B= \frac{ \pi} {8N}$ was used along with the lower bound $\sin(x)\geq x/2$ for $x \in [0, \pi/2]$. 
For the case $y<N/2$ we carry out a similar computation using Corollary \ref{cor:chb} to lower bound the relative entropy, 
\all{ 
m D( G(y- 0.25) || G(y)) &\geq m \frac{ (\cos^{2} ( A ) -  \cos^{2} ( B ))^{2}  } { 2\sin^{2} ( A ) } \nl 
&=  m \frac{ \sin^{2} (A+B) \sin^{2} (A-B)  } { 2\sin^{2} ( A ) } \nl 
&= \frac{m \sin^{2} (A-B) }{ 2} (\cos(B) + \sin(B)\cot A))^{2} \nl 
& \geq 50 c N^{2}  \sin^{2} (\frac{\pi}{8N}) > 2 c. 
} {none} 
For the last steps in the computation, the inequality $(\cos(B) + \sin(B)\cot (A))^{2}>1$ for $A<\pi/4$ and $A-B= \frac{ \pi} {8N}$ was used along with the lower bound $\sin(x)> x/2$  for $x \in [0, \pi/2]$.  

Similarly, $F(\widehat{p}) - F(p)  \geq 0.25$ is equivalent to $G(y+ 0.25) > \widehat{p}$ and the probability can be bounded using the additive Chernoff bound, 
\al{ 
\Pr [ \widehat{p} < G(y+0.25)   ] &\leq e^{ -m D( G(y+ 0.25) || G(y)) } 
} 
\noindent Further, $ mD( G(y+0.25) || G(y)) \geq 2c$ for all $y \in [0,N]$ with probability at least $1- e^{-2c}$, this follows from 
calculations similar to the ones in equations \eqref{eight} and \eqref{none} with denominators $(\cos^{2}(B), \sin^{2} (A))$ replaced by $(\cos^{2}(A), \sin^{2} (B))$. 
From the union bound, it follows that $|\widehat{l} - (-1)^{t} l \mod N | \leq 0.25$ with probability at least $1- 2e^{-2c}$. 

\end{proof} 
\noindent The remaining steps of the algorithm resolves the sign ambiguity in the estimates obtained in steps 10-11 by comparing against the recursive estimate. The following lemma bounds smallest possible difference between possible 
reconstructions produced by the Chinese Remainder Theorem up to sign ambiguities. 
\begin{lemma}\label{lemma:signs}
For an integer $0\leq M < N$, consider the Chinese remainder theorem bijection $M=\text{CRT}(m_1,\dots,m_K)$, where $M\equiv m_i\mod N_i$. Let $N_1<\dots <N_{K}$ be the ordered $K$ odd coprimes.

The only possible integers $M'\neq M$ such that $|M-M'|< \min N_i$ such that $M'$ can be obtained via CRT by changing some signs on the moduli, namely $M'=\text{CRT}(\pm m_1, \pm m_2, \dots, \pm m_k)$ are contained in an interval around $N/2$, 
\begin{equation}
	M' \in \left [ \frac{N-N_1}{2}+1,\,\frac{N+N_1}{2}-1\right ] \cap \mathbb{Z}\,.
\end{equation}
\end{lemma}
\begin{proof}
The CRT function is `continuous' in the sense that if $CRT( m_1, \dots, m_K ) = M$ then $CRT(m_1+a, \dots, m_K+a) = M+ a \mod N$ for any displacement $a$. We therefore examine the displacements of $-N_1+1\leq a \leq N_1-1$ around $M$ and seeing if they can be made to match with a different sign resolution. In other words, there should be such a choice of $a$ such that $m_i + a \equiv \pm m_i \mod N_i$ for all the moduli $N_{i}$ for $i\in[k]$. Since $|a|<N_1$ and $N_1$ can be assumed to be the smallest modulus, we see that this is not possible unless all the signs are flipped, i.e. $M'=\text{CRT}(-m_1,\dots,-m_K)$ and $m_i+a\equiv -m_i \mod N_i$, for all $i \in [K]$.

Second, since the $N_i$ are all odd, there are two possibilities for possible solutions $(M, M')$ (where $M$ is assumed to be larger) in terms of the parity of $a$:
\begin{itemize}
\item If $a$ is even, then $(M, M')= (N- a/2, a/2)$ are the unique solutions to the equations $2m_{i}+ a = 0 \mod N_{i}$. In this case, $|M'-M|=N-a \geq N-N_1 > N_1$, contradicting the hypothesis in the theorem statement. 
\item If $a$ is odd, then $(M, M')=((N+a)/2, (N-a)/2)$ are the unique solutions to the equations $2m_{i} + a = 0 \mod N_{i}$. For the QoPrime algorithm, this case is resolved by the depth $0$ measurements carried out in the recursive step.
\end{itemize}
\end{proof}

It is this second set of $N_1-1$ special points (out of the $N$ possibilities for $M$) which presents an issue to our algorithm from the point of view of resolving the sign ambiguity. These problematic angles lie in a small interval of $\theta\in\left[\frac{\pi(N-N_1)}{4N},\,\frac{\pi(N+N_1)}{4N}\right]$ around the midpoint $\pi/4$. We use the technique of exact amplitude amplification which allows us to map an angle away from a known interval, should it happen to be there. Note that the size of the problematic interval is of the order of the estimation accuracy for the recursive estimate in step 1-5, so only a small fraction of points need to be handled.

\begin{lemma}\label{lemma:exactAA}
[Exact amplitude amplification \cite{BHMT02}] If $\theta\in\left[\frac{\pi(N-N_1)}{4N},\,\frac{\pi(N+N_1)}{4N}\right]$, then appending an extra qubit in the state $(\ket{0} + \ket{1})/\sqrt{2}$ and selecting on $\ket{00}$ results in an oracle $U'\ket{0^t} = \cos(\theta') \ket{x, 00} +\sin(\theta') \ket{x', (00)^{\perp}}$ where $|\pi/4- \theta'| \geq 0.07\pi$ for $N_{1} \leq 0.04 N$. 
\end{lemma}
\begin{proof} 
If the angle $\theta=\pi/4$, then $\cos(\theta') = \cos(\theta)/\sqrt{2} = 1/2$ and thus $\theta' = \pi/3$ and the difference $|\pi/4- \theta'| \geq \pi/12\geq0.08 \pi$. The forbidden interval is 
asymptotically smaller than $[\pi/4 -0.01, \pi/4 + 0.01]$ for large enough $N$ (that is $N_{1}= O(N^{q/k}) \leq 0.04N$ for large enough $N$), the result follows. 
\end{proof}

We next describe the procedure in steps 15-17 of the QoPrime algorithm for estimating the values $M_i$ that will be used in the Chinese Remainder theorem. 
Define the confidence intervals $A_{i} = [ \{ \overline{M_{i}} \} - 0.25, \{ \overline{M_{i}} \}   + 0.25 ]$ corresponding to all the estimates $\{ \overline{M_{i}} \}$ produced 
by the QoPrime algorithm. Let $I = \bigcap_{i} A_{i}$ be the intersection of the $A_{i}$. Applying the union bound and Lemma \ref{lem:one} it follows that $I$ is non empty and the fractional part $\{ M \} \in I$ with probability at least $1- 2ke^{-c}$. Step 9 of the QoPrime algorithm is therefore able to find $\alpha \in I$. In Step 10, from $\alpha$ one finds $\beta_{i}  \in [-0.25, 0.25]$ such that $\{ \overline{M_{i}} + \beta_{i} \} = \alpha$ and then the value $M_i$ is computed as $M_{i} = \lfloor \overline{M_{i}} + \beta_{i}  \rfloor$. It remains to show that using the Chinese Remainder Theorem on these values produces an estimate with error $\epsilon$ with high probability.
\begin{theorem} \label{thm1} 
The estimate output by the QoPrime algorithm is within additive error $\epsilon$ of the true value with probability at least $1- 2ke^{-2c}$. 
\end{theorem} 
\begin{proof} 
Let $CRT(m_{1},\dots,m_{k})$ be the function in Definition \ref{def:crt} denote the unique integer $m \mod N$ such that $m =  m_{i} 
\mod N_{i} $. The Chinese remainder theorem shows that this function is invertible, specifically $CRT^{-1}(m) = (m \mod n_1, \dots, m \mod n_k)$. 
The $CRT$ function is continuous in the following sense $CRT^{-1}(m+a) = (m + a\mod n_1, \dots, m+a \mod n_k)$, or equivalently $CRT(m_1+a, m_2+a, \dots, m_k+a) = CRT(m_1,\dots, m_k)+a$ for $a \in \Z$. For $M = \lfloor M \rfloor + \{ M\}$, we have 
\al{ 
M= CRT (  \lfloor M \rfloor \mod N_{1}, \cdots, \lfloor M \rfloor \mod N_{K} ) + \{ M \} 
} 
The QoPrime algorithm instead outputs the reconstructed estimate, 
\al{ 
\overline{M}= CRT (  \lfloor \overline{M_{1}} + \beta_{1}  \rfloor , \cdots,   \lfloor \overline{M_{K}} + \beta_{i}  \rfloor) + \alpha
} 
Let us establish next the correctness of the sign resolution procedure. The reconstruction $\theta_{s}$ for the correct sign pattern is within distance $N_{1}/2N$ of the estimate $\theta'$ produced by Lemma \ref{lem:zero} with probability $1- 6e^{-2c}$. Lemma \ref{lemma:signs} establishes that no other sign pattern can have comparable accuracy except for the exceptional case when the $\theta$ is close to $\pi/4$ which is handled by the exact amplitude amplification technique in Lemma \ref{lemma:exactAA}. It follows that the sign resolution procedure in step 18 is correct works and it can be assumed for the analysis that the signs $t$ are correct.

By Lemma \ref{lem:one}  and the choice of $\beta_{j}$  in step 15 of the Algorithm it follows that  $| \overline{M_{j}} + \beta_{j}  - (\lfloor M \rfloor \mod N_{j} +  \{ M \}   )| = \gamma < 0.5 $ for all $j \in [k]$ where $\gamma$ is independent of $j$ with probability at least $1- 2ke^{-2c}$. The accuracy of the estimates implies that $| \lfloor \overline{M_{j}} + \beta_{j}  \rfloor - (\lfloor M \rfloor \mod N_{j}   )| \leq 1 $ for all $j \in [k]$. By continuity of the Chinese remainder theorem, the reconstruction error $|\overline{M} - M| \leq 1+ |\alpha - \{M\}| \leq 1.5$. The QoPrime algorithm therefore outputs an estimate $\frac{\pi \overline{M}  }{2N}$ that differs from the true value $\frac{\pi M  }{2N}$ by at most $\frac{\pi}{N} \leq \epsilon$.

\end{proof}

\subsection{Choosing the parameters}\label{sec33}

In this section, we further detail the choice of the parameters for the QoPrime algorithm. 
The small-$\epsilon$ asymptotics of the QoPrime algorithm rely on finding coprime moduli of similar magnitude and product of order $\Theta(\epsilon^{-1})$. To this effect, we formulate the following lemma:
\begin{lemma}
	\label{lemma:coprimes}
	\cite{ES71} Given a fixed integer $k\geq 2$ and $N \in \R$, we can find $k$ mutually coprime integers $n_1(N)< n_2(N)<\dots<n_k(N)$ such that $\lim_{N\to\infty}\frac{n_1(N)\dots n_k(N)}{N}=1$ and $\lim_{N\to\infty}\frac{n_k(N)}{n_1(N)}=1$.
\end{lemma}

This lemma follows from the sub-linear scaling of the number of coprimes that can fit inside an interval of a given size, as studied in \cite{ES71}. In practice, we pre-compute a table of $k$ adjacent odd coprimes starting at each odd integer, stopping at large enough values of $k$ and $n_1$ according to the target precision $\epsilon$. It suffices to build the table for $k \leq 12$ and $n_1\approx 10^5$ to achieve approximation error $\epsilon=10^{-10}$, both with and without noise. Given the table, target precision $\epsilon$ and an integer $k\geq 2$, the implementation chooses $k$ adjacent coprimes from the table starting at $\lfloor 1/\epsilon^{1/k} \rfloor$ with product closest to $\frac{\pi}{\epsilon}$ in absolute value. Figure 1 
compares the table used in practice to the theoretical guarantees in Lemma \ref{lemma:coprimes}.\\

\noindent We can summarize the asymptotics of the algorithm in the following table:
\begin{equation}
    \label{eq:noiselessalgo}
    \begin{array}{lll}
    	\text{parameters:} & \epsilon , k, q & \\
    	\text{coprimes:} & n_1,\dots, n_k &=\Theta(1/\epsilon^{1/k}) \\
	    \text{sampling depth:} & \text{max}_{i}\frac{N}{N_i} &= \Theta(1/\epsilon^{1-q/k})\\
    \text{oracle calls:} & O\left(\sum_{i=1}^{\lceil k/q \rceil} N_i^2 \times \frac{N}{N_i}\right) &= O\left(\left\lceil \frac{k}{q}\right\rceil 1/\epsilon^{1+q/k}\right)
    \end{array}
\end{equation}

\begin{figure}[H]
	\centering
	\includegraphics[scale=0.33]{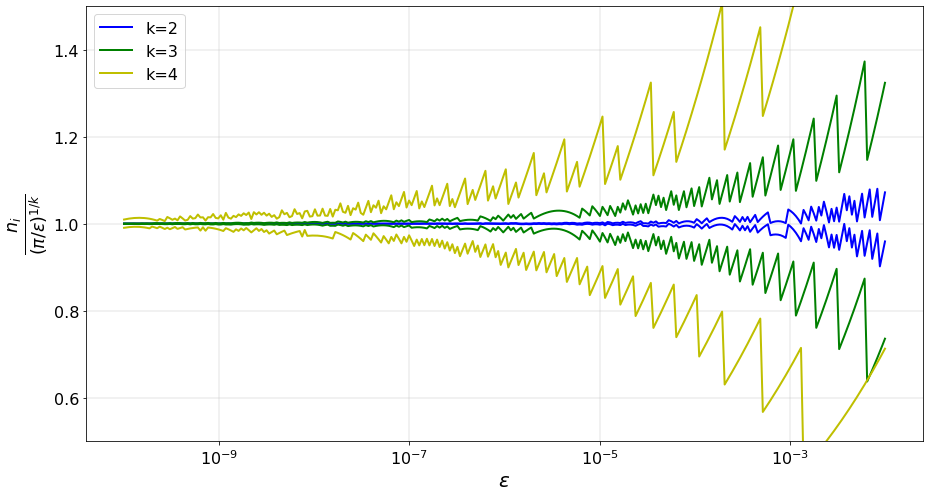}
	\caption{Convergence of the coprime finding routine. The algorithm finds $k$ adjacent coprimes $n_1$, $n_2$, \dots, $n_k$ such that their product $N=n_1\dots n_k$ is close to $\pi/\epsilon$. 
	Both the smallest coprime $n_1$ (approaching 1 from below) and the largest coprime $n_k$ (approaching 1 from above) are shown to converge to $(\pi/\epsilon)^{1/k}$ matching the 
	convergence in Lemma \ref{lemma:coprimes} for several values of $k$. }
\end{figure}\label{fig:coprimeconvergence}

Note that a bound onthe number of total classical sampling sessions used in the algorithm can be obtained easily using our geometric schedule described in Section 3.2. Since for every recursive call the ratio $q/k$ is divided by two, we have that for the $n$th recursive step there will be $\lceil k / 2^n q\rceil$ sampling sessions. Therefore the total number of sampling sessions can be bounded by the simple geometric sum $\lceil k / q\rceil + \lceil k / 2 q\rceil + \lceil k / 2^2 q\rceil + \dots \leq 2\lceil k/q\rceil$. While simple, this fact is important because it tells us that the complexity of the algorithm can be described only by the first level of the recursion, up to a small order-one factor. Therefore, in choosing the algorithm's parameters such as $k$ and $q$, we will only optimize over the behavior of the first level of the recursion. This analysis means that, for a given target precision $\epsilon$ and overall failure probability $\delta$, the total number of oracle calls can be bounded by:
\begin{equation}
	\label{eq:noiselessprefactor}
	\mathcal{N}(k, q, \epsilon, \delta) \leq  C \times \left\lceil\frac{k}{q}\right\rceil\times \frac{1}{\epsilon^{1+q/k}}\times \log\left(\frac{4}{\delta}\left\lceil\frac{k}{q}\right\rceil\right)
\end{equation}
The constant $C$ which tightens this bound does not depend on the parameters of the algorithm, in fact experimental results provide evidence that $C$ is a small constant and that in practice, we can expect $C < 10$. (see Figure \ref{fig:prefactor}). 

In practice, given target precision $\epsilon$ and accepted failure probability $\delta$, we can find optimal values of $k$ and $q$ such that the oracle calls in (\ref{eq:noiselessprefactor}) is minimized. In this noiseless scenario, it can be shown that the optimal parameter $q$ is always 1 (i.e. it is best to access the largest allowed depth). Minimizing the oracle call number in (\ref{eq:noiselessprefactor}) by choosing $k$ leads to (ignoring the subleading contribution from the failure probability $\delta$):
\begin{equation}
	k^*(\epsilon)\approx \log\frac{1}{\epsilon}
\end{equation}
Plugging this back into (\ref{eq:noiselessprefactor}) leads to an asymptotic dependency of oracle calls of $1/\epsilon$, up to logarithmic factors, as expected in the quantum regime:
\begin{equation}
	\label{eq:noiselesslimit}
	\min_{(q,k)} \mathcal{N} = O\left(\epsilon^{-1}\log\epsilon^{-1}\log\left(\frac{4 \log\epsilon^{-1}}{\delta}\right)\right)
\end{equation}


\subsection{QoPrime algorithm with noise} 

In this section, we study the performance of the QoPrime AE algorithm under the same depolarizing noise model as we studied for the power law AE algorithm. Inverting the noisy probabilistic model (\ref{eq:depoldef}), the classical inference problem in the algorithm, when sampling at depth $N/N_i$, becomes:
\begin{equation}
		\overline{M_j} \equiv \pm \frac{2N_i}{\pi}\arccos\sqrt{\frac{1}{2}+e^{\gamma N/N_i}\left(\hat{p}-\frac{1}{2}\right)}\mod N_i
\end{equation}
As before, we can use this relation to translate a confidence interval on the coin toss probability $\hat p$, computed by classical postprocessing of measurement samples, to a confidence interval on $\overline{M_j}$. The difference in the noisy case is the exponential stretch factor of $e^{\gamma N/N_i}$ enhancing the angle confidence interval. Since a classical confidence interval shrinks as the square root of the number of samples, the required number of samples will pick up a factor of  $e^{2 \gamma N/N_i}$ under this noise model in order to guarantee the noiseless confidence intervals. Similar to the noiseless algorithm (\ref{eq:noiselessalgo}), we can therefore summarize the asymptotics of the noisy algorithm as follows:

\begin{equation}
    \label{eq:noisyalgo}
    \begin{array}{lll}
    \text{parameters:} & \epsilon , k, q, \gamma & \\
    	\text{coprimes:} & n_1,\dots, n_k &=\Theta(1/\epsilon^{1/k}) \\
	    \text{sampling depth:} & \text{max}_{i}\frac{N}{N_i}  &= \Theta(1/\epsilon^{1-q/k})\\
    \text{oracle calls:} & \Theta\left(\sum_{i=1}^{\lceil k/q \rceil} N_i^2 e^{2\gamma N/N_i}\frac{N}{N_i}\right) &= \Theta\left(\left\lceil \frac{k}{q}\right\rceil \frac{1}{\epsilon^{1+q/k}}e^{2\gamma (\frac{\pi}{2\epsilon})^{1-q/k}}\right)
    \end{array}
\end{equation}
Therefore, for a given target precision $\epsilon$, depolarizing noise level $\gamma$, and overall failure probability $\delta$, the total number of oracle calls scales as:
\begin{equation}
	\label{eq:prefactor}
	\mathcal{N}(k, q, \gamma, \epsilon, \delta) \leq C \times \left\lceil\frac{k}{q}\right\rceil\times \frac{1}{\epsilon^{1+q/k}} \times\exp\left(2\gamma\left(\frac{\pi}{2\epsilon}\right)^{1-q/k}\right)\times \log\left(\frac{4}{\delta}\left\lceil\frac{k}{q}\right\rceil\right)
\end{equation}
where $C$ is the same constant overhead as in equation (\ref{eq:noiselessprefactor}). Given target precision $\epsilon$, noise level $\gamma$, and accepted failure probability $\delta$, we can find optimal values of $k$ and $q$ such that the number of oracle calls in (\ref{eq:prefactor}) is minimized. See also Figure 2 for the behaviour of the number of oracle calls for different values of $k$ and $q$.
\begin{figure}[H]
\label{fig:envelope}
\centering
\includegraphics[scale=0.33]{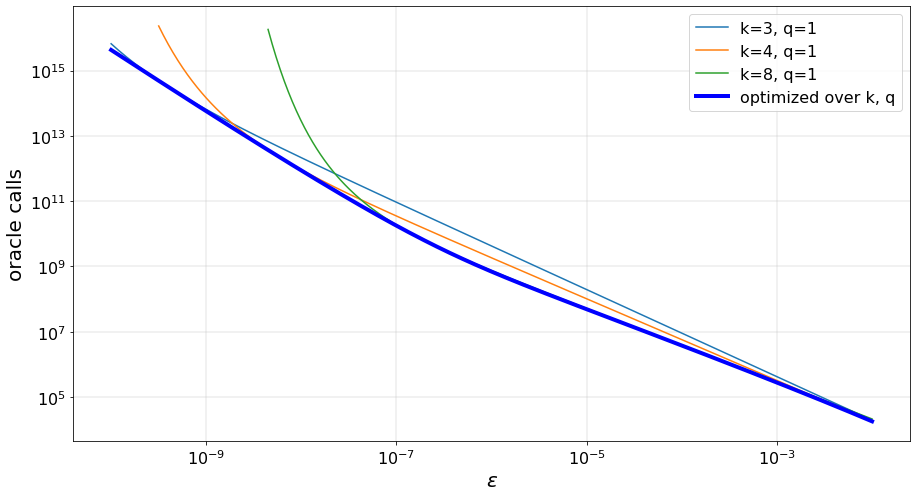}
\caption{The behavior of the oracle call dependency in (\ref{eq:prefactor}) for several values of parameters $k$ and $q$, and for fixed noise level $\gamma=10^{-5}$ and probability of failure $\delta=10^{-5}$. When taking the minimum over the family of curves parametrized by all valid $k$ and $q$ (here assumed continuous for simplicity), we obtain the envelope of the optimal QoPrime algorithm (thick blue line). This emergent behavior smoothly interpolates between a classical $1/\epsilon^{2}$ scaling in the noise-dominated region $\epsilon \ll \gamma$ and a quantum scaling $1/\epsilon$ in the coherent region $\epsilon \gg \gamma$.}
\end{figure}

Optimizing over $k$ and $q$ in this manner is the step which ensures that the effective scaling of oracle calls as a function of $\epsilon$ is always between the classical scaling of $1/\epsilon^{2}$ and the quantum scaling $1/\epsilon$ (see Figure 3). Specifically, this can be formulated as a bound on the instantaneous exponent: 
\begin{equation}
	-2 \leq \lim_{\epsilon \to 0}\frac{d \inf_{k,q}\log \mathcal{N}(\epsilon, \gamma, \delta, k, q)}{d\log\epsilon} \leq -1
\end{equation}

\begin{figure}[H]
\label{fig:exponent}
\centering
\includegraphics[scale=0.33]{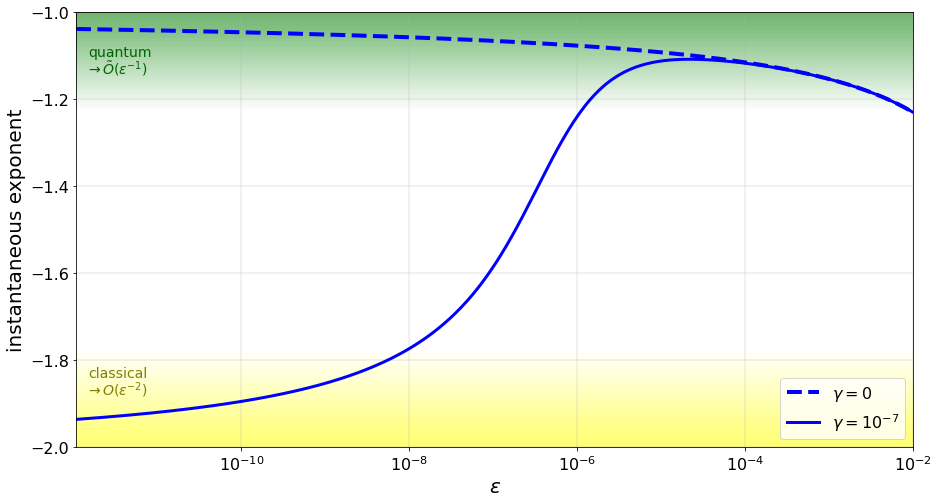}
\caption{The transition from coherent to noise-dominated as measured by the instantaneous exponent $\frac{d\log \mathcal{N}}{d\log \epsilon}$, where $\mathcal{N}$ is the number oracle calls optimized over parameters $k$ and $q$.}
\end{figure}

It can be shown that in the noise-dominated limit, the optimal parameter $q$ tends to its upper bound $q=k-1$ (corresponding to the shallowest accessible circuits). Using this observation, we can analytically study the optimization over $k$ using the dependency in (\ref{eq:prefactor}) and obtain that the optimal $k$ parameter will have the form (in a continuous approximation):
\begin{equation}
	k^*(\epsilon, \gamma) \approx \frac{\log\frac{\pi}{2\epsilon}}{\log\frac{\log\frac{1}{\epsilon}}{2\gamma\log\frac{\pi}{2\epsilon}}}
\end{equation}
This allows us to study the asymptotic dependency of oracle calls on the target precision $\epsilon$ analytically; extracting the low-$\epsilon$ limit yields:
\begin{equation}
	\label{eq:noisylimit}
	\lim_{\gamma \gg \epsilon} \mathcal{N}(\epsilon, \gamma, \delta) \leq 4Ce\,\gamma\, \epsilon^{-2} \log\left(\frac{4}{\delta}\right)\,,
\end{equation}
where $C$ is the constant prefactor introduced in (\ref{eq:prefactor}). This classical-limit curve can be used to compare the asymptotic runtime of our algorithm to classical Monte Carlo techniques.
\\\\
Outside of the two noise limits, specifically noise-dominated (\ref{eq:noisylimit}) and noiseless (\ref{eq:noiselesslimit}), the optimal parameters $k$ and $q$ depend non-trivially on the problem, and they can be found numerically. Example $(k, q)$ optimal trajectories obtained by optimizing (\ref{eq:prefactor}) are shown in Figure 4.
\begin{figure}[H]
\label{fig:parameter_trajectory}
\centering
\includegraphics[scale=0.33]{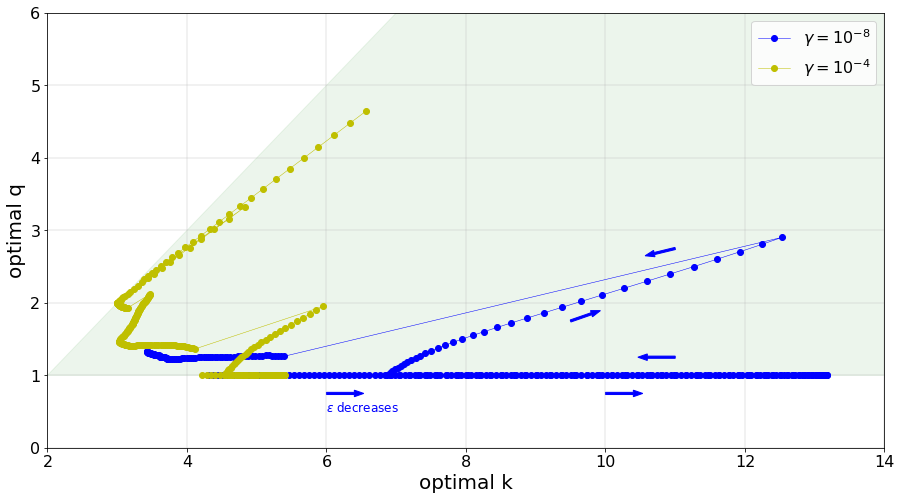}
\caption{The trajectory of optimal $k$, $q$ parameters chosen by minimizing the asymptotic dependency in (\ref{eq:prefactor}), for two different noise levels. The optimization is over continuous $k$, $q$ for simplicity. The green region marks the valid parameter region $k\geq 2$, $1 \leq q \leq k-1$. The arrow shows the direction of the optimal parameters as the target precision $\epsilon$ is being lowered from $\epsilon=10^{-3}$ to $\epsilon=10^{-10}$. While $\epsilon \ll \gamma$ (i.e. noiseless regime), we have that $q=1$.}
\end{figure}


\section{Empirical results}\label{empirics}
In this section, we present empirical results for the power law and the QoPrime AE algorithms and compare them with state of the art amplitude estimation 
algorithms \cite{grinko2019iterative}. The experimental results validate the theoretical analysis and provide further insight in the behaviour of these low depth algorithms in noisy regimes. 

\subsection{The power law AE}

Figure 5 compares the theoretical and empirically observed scaling for the number of oracle calls $N$ as a function of the error rate $\epsilon$ in the power law AE algorithm in the absence of noise, i.e. $\gamma = 0$. We numerically simulated the power law AE algorithm for randomly chosen $\theta  \in [0, \pi/2]$ and with $m_k=\lfloor k^{\frac{1-\beta}{2\beta}} \rfloor$ for fixed values of parameter $\beta \in \{0.455, 0.714 \}$, which make the exponent be $\{0.2,0.6\}$ respectively. We also provide the extremal cases of $\beta \in \{0,1\}$. The experimental results agree closely with the predictions of the theoretical analysis. 

Figure 6 shows the scaling of the power law AE algorithm under several noise levels. Here, the parameter $\beta$ is chosen adaptively, based on the error rate $\epsilon$ and noise level $\gamma$. For target errors below the noise level, we can use the exponential schedule to get the optimal quantum scaling. After this threshold, we use the power law schedules with exponents chosen as in Proposition \ref{prop:pwl_exponent_noise}. The result is that for these smaller target errors, the scaling is in between the optimal quantum and the classical scaling.

\begin{figure} [H] \label{f7} 
\centering
\includegraphics[scale = 0.28]{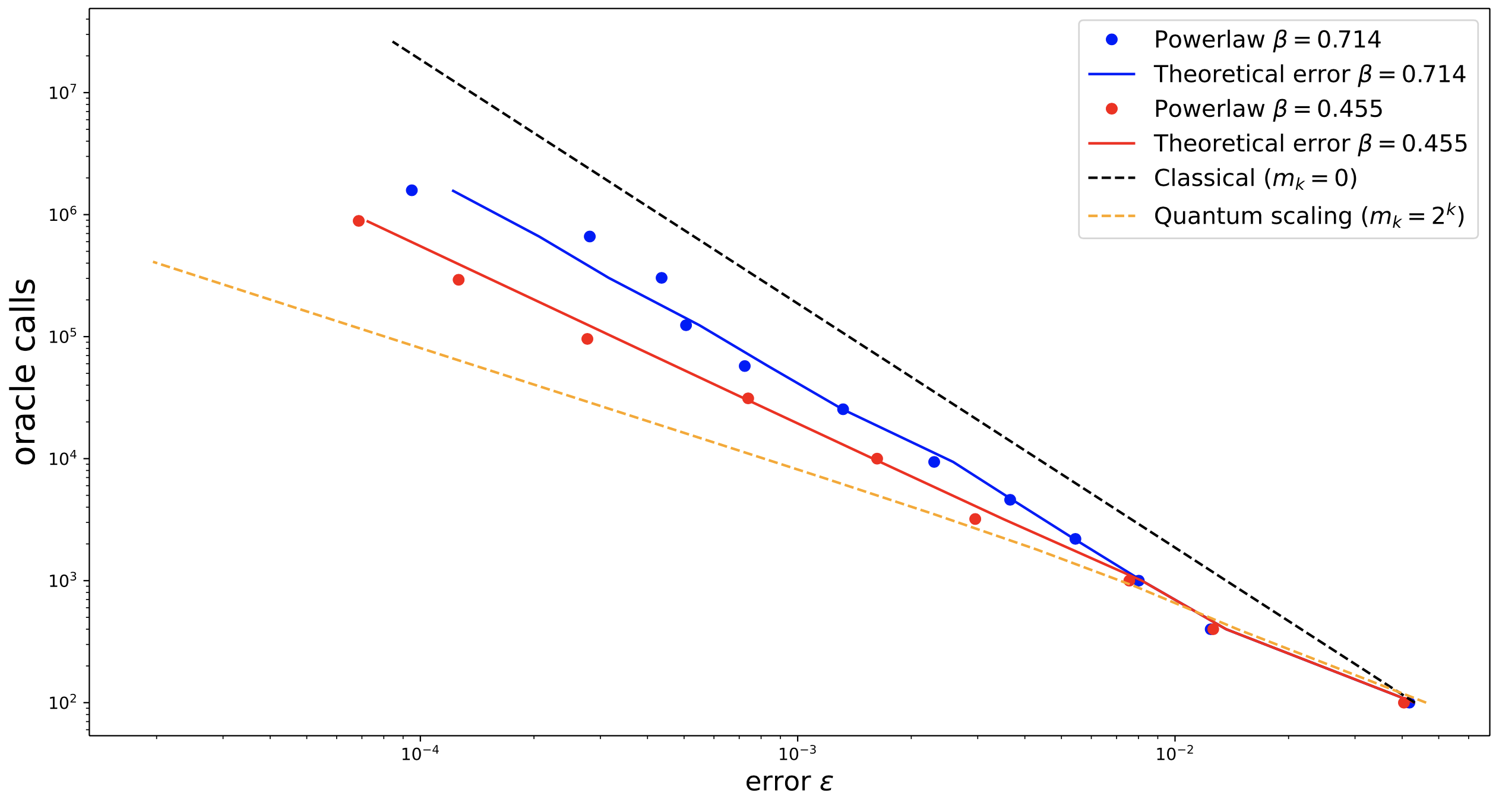}
\caption{Performance of the power law AE algorithm in theory (solid) and practice (dots) using the schedule $m_k=\lfloor k^{\frac{1-\beta}{2\beta}} \rfloor$ for values $\beta=0.455$ (red) and $\beta=0.714$ (blue), for different error rates $\epsilon$. The true value is $\theta^*$ is chosen at random. We took the number of shots $N_{shot}=100$. Applying linear regression to these experimental data points, gives slopes $-1.718$ and $-1.469$, whereas the theoretical slopes are $-1.714$ and $-1.455$ for the blue and red points respectively. }
\end{figure}

\begin{figure}[H]
\centering
\includegraphics[scale=0.28]{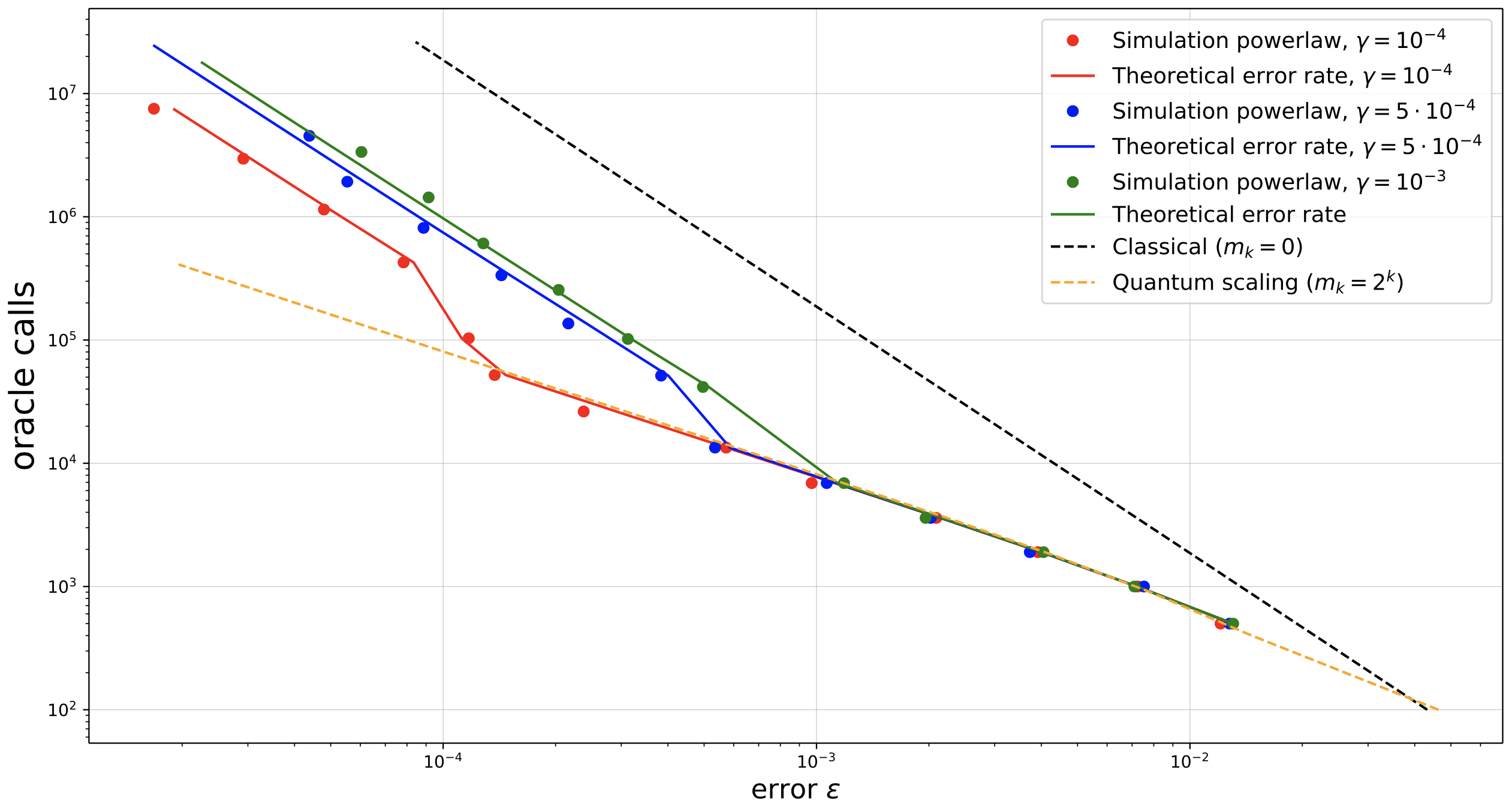}
\caption{Performance of the power law AE algorithm in theory (solid) and practice (dots) using power law schedules where the parameter $\beta$ is optimized for given $\epsilon$ and $\gamma$. Also the classical and quantum scalings are plotted for comparison. For small target errors, we obtain the optimal quantum scaling, while for smaller target errors we use power law exponents using Proposition \ref{prop:pwl_exponent_noise}. The result is that the scaling approaches the classical scaling as the target error goes to 0.}
\end{figure}

\subsection{The QoPrime algorithm}

The parameters $k$ and $q$ for the QoPrime algorithm are chosen by optimizing over the Chernoff upper-bounds obtained in Lemma \ref{lem:one} and described in (\ref{eq:noisyalgo}). Figure 7 shows the theoretical upper bounds and the empirically observed number of oracle calls as a function of the accuracy $\epsilon$ for different noise rates. The algorithm in practice performs better than the theoretical bounds as it computes the confidence intervals using exact binomial distributions as opposed to the Chernoff bounds in the theoretical analysis.

Figure 8 plots the maximum oracle depth as a function of the target precision $\epsilon$ for the QoPrime algorithm in noiseless and noisy settings, as well as for the IQAE algorithm. Finally, Figure 9 provides empirical estimates for the constant factor $C$ for the QoPrime algorithm in noisy settings. The observed value of $C$ is a small constant and the simulations show that $C <10$ over a wide range of $\epsilon$ and noise rates that cover most settings of interest.


\begin{figure}[H] \label{f8} 
\centering
\includegraphics[scale=0.33]{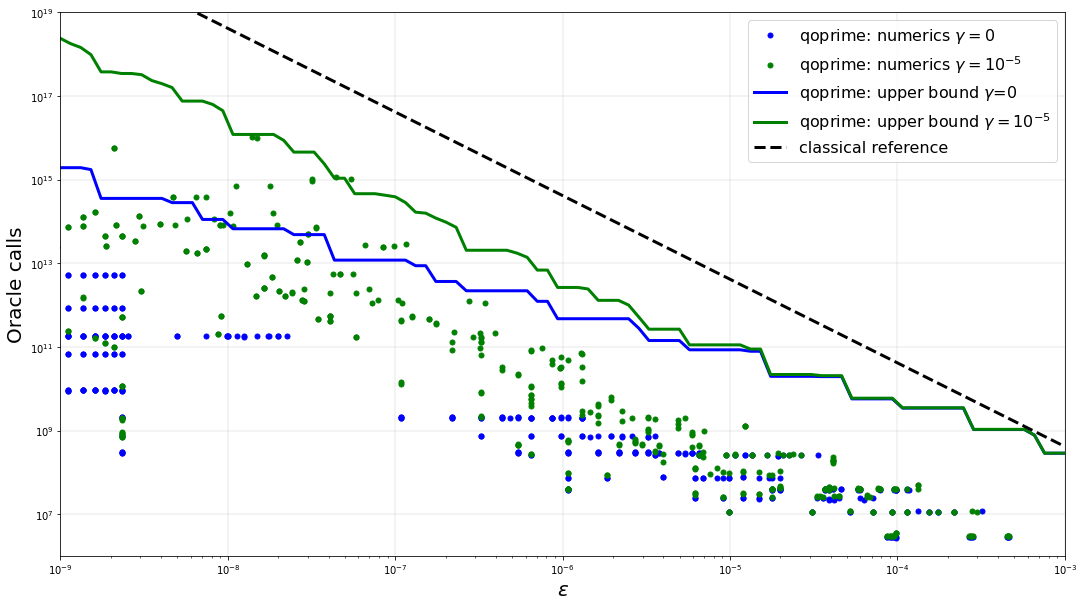}
\caption{Performance of the QoPrime algorithm under two noise levels, and across various choices for the true angle $\theta$. Shown are both theoretical upper bounds (solid) and exact simulated oracle calls (dots) for two scenarios: noiseless (blue), and depolarizing rate $\gamma=10^{-5}$ (green). For each target precision, 20 values of the true angle spanning the $[0, \pi/2]$ interval have been selected to generate the samples; the horizontal axis represents realized approximation precision. The classical Monte Carlo curve (black) is obtained by assuming noiseless classical sampling from a constant oracle depth of 1. We see the curve follow a quantum $\epsilon^{-1}$ scaling for small errors $\epsilon \gg \gamma$, which transitions into a classical $\epsilon^{-2}$ dependency when the precision is much smaller than the noise level ($\epsilon \ll \gamma$).}
\end{figure}

\begin{figure}[H] \label{f9} 
\centering
\includegraphics[scale=0.425]{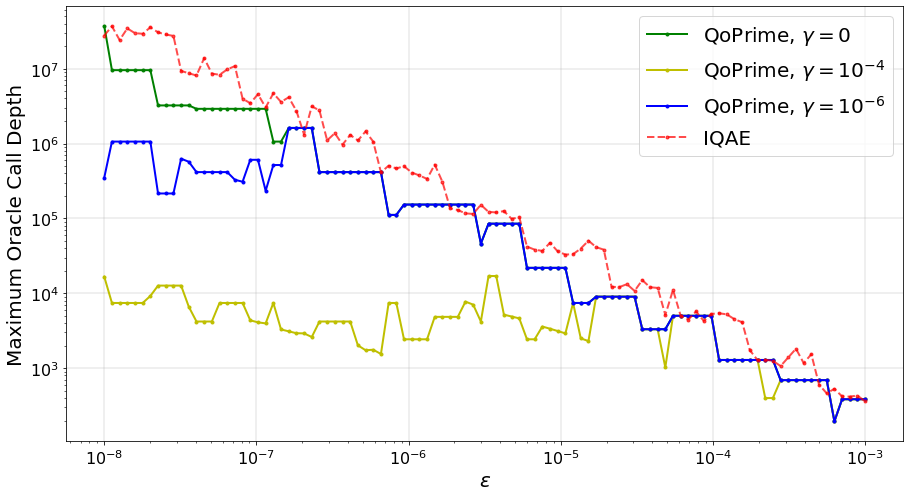}
\caption{Maximum oracle depth as a function of the target precision $\epsilon$. The noiseless QoPrime algorithm and the IQAE algorithm in \cite{grinko2019iterative} both have a similar scaling of depth as $O(\epsilon^{-1})$. However, introducing a depolarizing noise level $\gamma$ provides a bound for the depth required by the QoPrime algorithm. Specifically, the QoPrime algorithm will not access depths higher than the noise scale $\gamma$, which would correspond to exponentially suppressed confidence intervals, and require exponentially more samples.}
\end{figure}

\begin{figure}[H]

\centering
\includegraphics[scale=0.30]{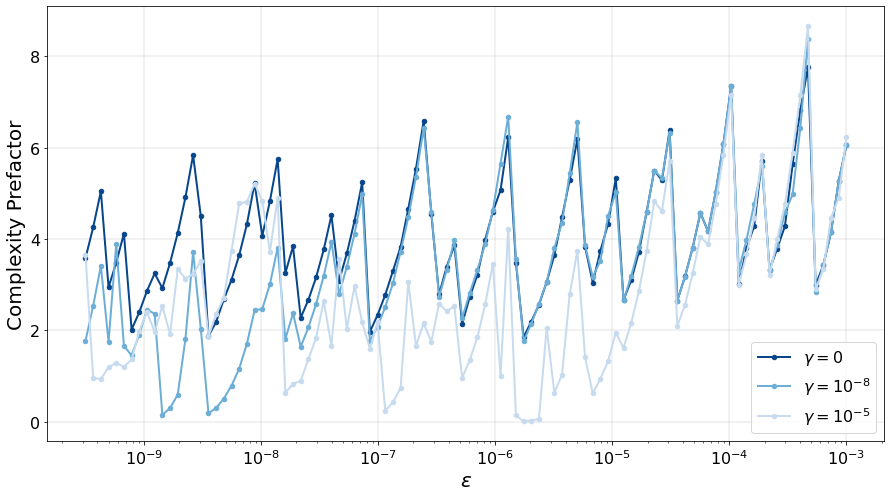}
\vspace{-10pt}
\caption{Empirical values of the constant prefactor $C$, as defined in (\ref{eq:prefactor}) above, on the target precision $\epsilon$ for an arbitrary value of the true angle $\theta$ and failure probability $\delta=10^{-5}$.} \label{fig:prefactor} 
\end{figure}

\subsection{Benchmarking}

Last we compare the performance of the Power law and QoPrime AE algorithms against the state of the art amplitude estimation algorithm IQAE  \cite{grinko2019iterative}. 

Figure 11 plots the performance of the Power Law, the QoPrime and the IQAE in noisy settings, where performance is measured by the number of oracle calls for target accuracy $\epsilon$. The plot emphasizes the advantage of the Power law and the QoPrime over algorithms such as IQAE, which require access to a full circuit depth of $O(1/\epsilon)$. In this scenario, this large depth is exponentially penalized by the depolarizing noise by requiring an exponentially large number of classical samples to achieve a precision below the noise level. In comparison, the power law and the QoPrime AE algorithms transition smoothly to a classical estimation scaling and do not suffer from an exponential growth in oracle calls. The Power law AE algorithm has the best practical performance according to the simulations.

\begin{figure}[H] \label{f10} 
\centering
\includegraphics[scale=0.28]{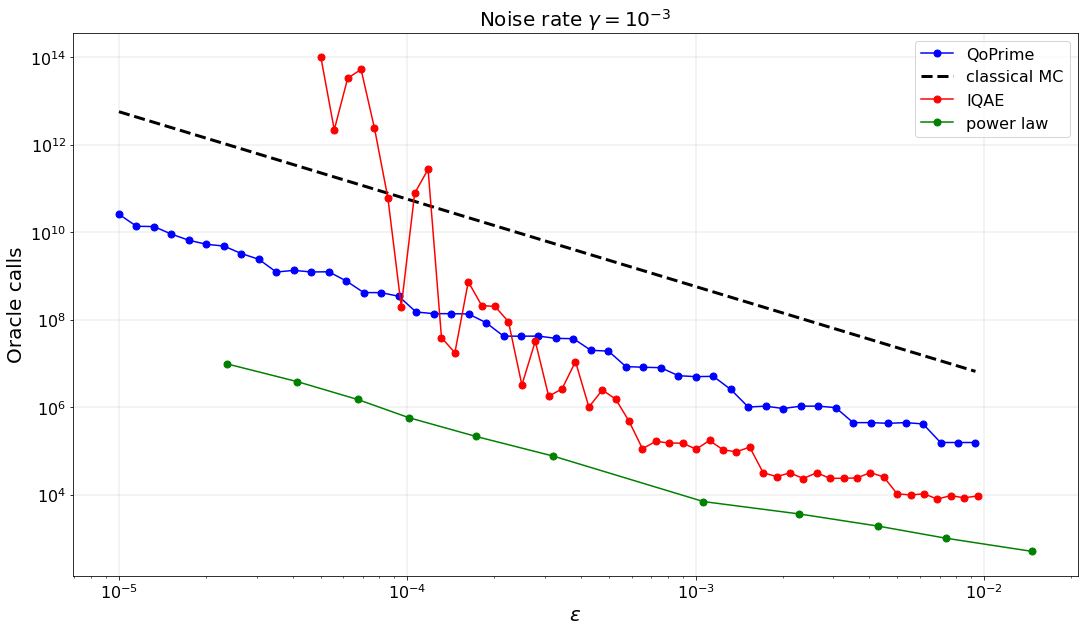}
\vspace{-10pt}
\caption{Comparison of the two algorithms introduced in this work (Power law and QoPrime) against the Iterative Quantum Amplitude Estimation algorithm (IQAE) introduced in \cite{grinko2019iterative}. A noise level of $\gamma = 10^{-3}$ is used for all three.}
\end{figure}

\textbf{Acknowledgements:} This work is a collaboration between Goldman Sachs and QCWare, it was carried out during TGT and FL's internships at Goldman Sachs. We acknowledge helpful discussions with Adam Bouland, Nikitas Stamatopolous, Rajiv Krishnakumar, and Paul Burchard. We thank an anonymous reviewer for helpful comments and for pointing out errors in a previous version of the QoPrime algorithm.


\bibliographystyle{IEEEtranS} 
\bibliography{bibliography}

\appendix

\section{Appendix: Regularity conditions for Bernstein Von-Mises Theorem} \label{BvM} 
We enumerate the regularity conditions for the Bernstein Von Mises theorem (Section 4, [HM75]) for power law schedules. Let us consider a schedule where $k$ oracle calls in series are made $N_{k}$ times followed by measurements in the standard basis and Bayesian updates. The random variable $N_{k_0}$ and $N_{k_1}$ represent the number of times outcomes $0$ and $1$ are observed out of the $N_{k}$ measurements. The probability density function and the log likelihood are given by, 
\al{ 
f(X, \theta) &= \frac{1}{Z} \prod_{k} \cos^{2}((2k+1)\theta)^{N_{k_0}} \sin^{2} ((2k+1) \theta)^{N_{k_1}} \\
l(X, \theta) &= \sum_{k} 2N_{k_0} \log  \cos((2k+1)\theta) + 2N_{k_1} \log  \sin ((2k+1) \theta)  + \log Z 
} 
The regularity conditions state that there is a suitable domain $\Theta$ such that the following statements hold for all possible values of the measurement outcomes $X$ and for some integer $s\geq 2$. 
\begin{enumerate} 
\item $f(X, \theta)$ is continuous on $\Theta$. 
\item $l(X, \theta)$ is continuous on $\overline{\Theta}$. 
\item  For every $\theta \in \Theta$ there is an open neighborhood $U_{\theta}$ such that for $\sigma, \tau \in U$ we have $E_{\sigma} (l(X, \tau)^s)$ is bounded. (more precisely, 
 the supremum of $E_{\sigma} (l(X, \tau)^s)$ is finite). 
\item For every $(\theta, \tau) \in (\Theta, \overline{\Theta})$ there exist neighborhoods $U$ and $V$ of $\theta, \tau$ (these neighborhoods depend on $\theta$ and $\tau$) 
such that  the supremum  $E_{\sigma} | \inf_{\delta \in V} l(X, \delta)|^{s}$ over all $\sigma \in U$ is finite. 
 
\item $l(X, \theta)$ is twice differentiable on $\Theta$. 
\item For all $\theta \in \Theta$ there is a neighborhood $U_{\theta}$ such that  for all $\tau \in U_{\theta}$, 
\al{ 
0 < E_{\tau} [ l''(X, \tau)^{s} ]  < \infty 
 } 
This is the $s$-th moment of the Fisher information. 
 \item There are neighborhoods $U$ for all $\theta$ and a bounded function $k_{\theta}: X \to \R$ such that, 
 \al{ 
 \norm{ l''(X, \tau) - l''(X, \sigma) }  \leq \norm{ \tau - \sigma} k_{\theta} (X) 
 } 
 for all $\tau, \sigma \in U$. 
 \item The prior probability $\lambda$ is positive on $\Theta$ and $0$ on $\R \setminus \Theta$. 
 \item For every $\theta \in \Theta$ there is a neighborhood $U_{\theta}$ such that for all $\sigma, \tau \in U_{\theta}$ and constant $c_{\theta}>0$ such that, 
 \al{ 
 | \log \lambda(\sigma) - \log \lambda(\tau) | \leq \norm{ \sigma - \tau} c_{\theta}
 } 
 This is stated as being equivalent to the continuity of $\lambda'$ on $\Theta$. 

\end{enumerate} 
These regularity conditions can be sub-divided into three groups as follows: 
\begin{enumerate} 
\item Conditions 1-4 are about the smoothness of $f(X, \theta)$ and $l(X, \theta)$, they will be satisfied if the 
the norm of log-likelihood is bounded on $\Theta$. We can choose $\Theta$ to be a subinterval around the true value for which the log-likelihood is bounded.


\item Conditions 5-7 are about the smoothness of the Fisher information on $\Theta$. They assert that the Fisher information is bounded on $\Theta$ and is differentiable, this 
means that the log-likelihood function should have derivatives of order at least $3$. 

\item Conditions 8-9 are about the smoothness of the prior distribution, namely that the first derivative of the prior should be a continuous function. 
These are trivially true for the uniform distribution. 
\end{enumerate} 

Figure 1 in Section 2 illustrates that the log-likelihood function is smooth over a neighborhood of the true value indicating that the regularity conditions for the Bernstein-Von Mises theorem are plausible in this setting, for a large neighborhood $\Theta$ to be  around the true value. Algorithm \ref{AE:pls} is stated with $\Theta$ as the entire $[0, \pi/2]$ interval as this choice seems to work in practice, one can also imagine a slightly modified algorithm where the first few sampling rounds are used to get a rough estimate for the true value lying in a large interval $\Theta$ and for subsequent rounds the prior is uniform on $\Theta$, with convergence established using the Bernstein Von-Mises theorem. Adding noise may further regularize the log-likelihood functions and enforce the regularity conditions required for the Bernstein-Von Mises theorem. 


\end{document}